\documentclass[12pt]{article}
\usepackage{xcolor,colortbl}
\definecolor{green}{rgb}{0.1,0.1,0.1}

\usepackage{graphics,epsfig}
\usepackage[all]{xy}
\usepackage[english,activeacute]{babel}
\usepackage[]{graphicx}
\usepackage{latexsym}
\usepackage{amsmath, amsthm, amsfonts, amssymb}
\usepackage{verbatim}
\usepackage{color}
\usepackage{physics}
\usepackage{amssymb}

\usepackage{hyperref}
\usepackage{array}
\usepackage{multirow}
\usepackage{diagbox}
\usepackage{caption}
\usepackage{float}
\usepackage{subcaption}
\newtheorem{theorem}{Theorem}

\begin{document}

\newtheorem{theo}{Theorem}[section]
\newtheorem{definition}[theo]{Definition}
\newtheorem{lem}[theo]{Lemma}
\newtheorem{prop}[theo]{Proposition}
\newtheorem{coro}[theo]{Corollary}
\newtheorem{exam}[theo]{Example}
\newtheorem{rema}[theo]{Remark}
\newtheorem{remark}[theo]{Remark}
\newtheorem{corollary}[theo]{Corollary}
\newtheorem{example}[theo]{Example}
\newtheorem{principle}[theo]{Principle}
\newcommand{\ninv}{\mathord{\sim}}
\newtheorem{axiom}[theo]{Axiom}

\title{Measure-theoretic approach to negative probabilities}
\author{ Elisa Monchietti$^{1}$, César Massri$^{2}$, Acacio de Barros$^{3}$ and Federico Holik$^{4}$}


\maketitle

\begin{center}
\begin{small}
1 - Universidad Nacional de Rosario.\\
2 - Instituto de Investigaciones Matem\'{a}ticas "Luis A. Santalo".\\
3 - School of Liberal Studies, San Francisco State University, 1900 Holloway Ave., San Francisco, California, USA.\\
4 - Instituto de Física La Plata (CONICET-UNLP), Calle 113 entre 64 y 64 S/N, 1900, La Plata, Buenos Aires, Argentina.
\end{small}
\end{center}

\begin{abstract}
In this work, we elaborate on a measure-theoretic approach to negative probabilities. We study a natural notion of contextuality measure and characterize its main properties. Then, we apply this measure to relevant examples of quantum physics. In particular, we study the role played by contextuality in quantum computing circuits.
\end{abstract}

\section{Introduction}

As reported by Peter Shor \footnote{The video of P. Shor can be found in the following link: \href{https://www.youtube.com/watch?v=6qD9XElTpCE}{https://www.youtube.com/watch?v=6qD9XElTpCE}}, R. P. Feynmann was very interested in \emph{negative probabilities}. He thought that, perhaps, they could be used to give a natural explanation to the violation of Bell inequalities by quantum systems. In a subsequent paper, Feynmann gave arguments for considering negative probabilities as an interesting option for handling different problems of modern physics \cite{feynman_negative_1987}. They also called the attention of P.A.M. Dirac \cite{Dirac-1942}.

Negative probabilities play indeed a key role in many areas of quantum physics \cite{Hillery-PhysicsReports}. The most important application is, perhaps, the use of the Wigner function \cite{wigner_quantum_1932} in quantum optics problems \cite{Glauber-Negative} (for example, in the problem of quantum state estimation \cite{QuantumTomography-DiscreteWigner} and the determination of quantum correlations and classicality of quantum states \cite{Kenfack_2004,Cormik-ClassicalityDiscreteWigner,Deleglise2008}). More recently, negative probabilities have gained much interest in quantum information theory \cite{Ferrie_2011}, especially, after the suggestion that quantum contextuality could be the reason behind the speed-up of quantum algorithms \cite{howard_contextuality_2014} (see also  \cite{Veitch_2012,Galvao-WignerSpeedUp,Booth_Contextuality_and_Negative,shahandeh2021,Mari_Eisert_2012}). Indeed, the conection between quantum contextuality, no-signal models and negative probabilities has been studied with great detail \cite{Abramsky_2011} (see also \cite{de_barros_measuring_2015,kujala2019measures} for the use of negative probabilities as contextuality measures). In general, one can say that negative probabilities are used to characterize different features of quantum mechanics \cite{Spekkens-NegativityandContextuality-PRL,Singer-Stulpe-PhaseSpaceRepresentations}, specially, the no-signal condition \cite{SimulatingNon-signaling}.

Here we elaborate on a previous work \cite{Indistinguihsability_and_Negative_Probabilities} and study with great detail a \emph{measure theoretic approach} to negative probabilities. The advantage of this approach is that it is based in measure theory allowing to include infinite dimensional models very naturally, and in a way which is very similar to that of Kolmogorov. As such, it is the natural generalization of Kolmogorov's approach, the main difference being that it incorporates the notion of \emph{measurement context} from the very beginning. Differently from previous approaches, it does not rely on any Hilbert space structure. Thus, it is very suited for studying generalized probabilistic theories and contextuality scenarios in an operational way. Contextuality is naturally represented as the non-existence of a global positive probability distribution. This can be used to define a contextuality measure that that can be applied to quantum physics and more general (no-signal) probabilistic models. After reviewing the generalities of the Wigner function in Section \ref{s:Review}, we delve into the details of the definition presented in \cite{Indistinguihsability_and_Negative_Probabilities} in Section \ref{s:Main_Definitions}. Then, we prove our main results, that are separated in two parts. First, we analyze the problems of existence and uniqueness of our measures for situations that are relevant in physics and general probabilistic models in Section \ref{s:Properties}. Then, we turn to some applications in physics in Section \ref{s:Applications}. In particular, we include an application of the measure theoretic negative probabilities to the study of quantum circuits in connection to quantum contextuality. Finally, in section \ref{s:Conclusions} we draw some conclusions.

\section{Wigner distribution}\label{s:Review}

In order to illustrate better the idea of what a negative probability is, in this section we review the Wigner distribution (we follow \cite{Peres_Book} and \cite{Hillery1997}). It can be written as:

\begin{equation} \label{wigner}
    W(\textbf{r},\textbf{p}) = \frac{1}{(2\pi)^3}\int \psi^{*}\left(\textbf{r}+\frac{\hbar}{2}\textbf{s}\right)
    \psi\left(\textbf{r}-\frac{\hbar}{2}\textbf{s}\right)
    e^{i\textbf{p}\cdot \textbf{s}} d^3\textbf{s},
\end{equation}

\noindent and it is a useful tool in quantum optics. When performing the usual integrations, we obtain the marginal distributions:

\begin{equation} \label{eq:marginal1}
    \int W(\textbf{r},\textbf{p}) d^3 \textbf{p} = |\psi(\textbf{r})|^2.
\end{equation}

\noindent and

\begin{equation} \label{eq:marginal2}
    \int W(\textbf{r},\textbf{p}) d^3 \textbf{r} = |\psi(\textbf{p})|^2.
\end{equation}

Wigner functions of orthogonal states satisfy

\begin{equation}
    \int W_1(\textbf{r},\textbf{p})W_2(\textbf{r},\textbf{p})d^3\textbf{r}d^3\textbf{p} = 0
\end{equation}

This shows that Wigner functions may occasionally be negative and cannot be interpreted as probability distributions (the term ``quasiprobablity'' is sometimes used). Nevertheless, Wigner functions may give a qualitative feeling of the approximate location of a quantum system in phase space. They are often used to visualize the dynamical behavior of quantum systems.

Wigner functions are normalized by $\int W(\textbf{r},\textbf{p})d^3\textbf{r}d^3\textbf{p} = 1$, but they cannot be arbitrarily narrow and high, since they must also satisfy

\begin{equation}
    \int [W(\textbf{q},\textbf{p})]^2d^3\textbf{q}d^3\textbf{p}\leq (2\pi\hbar)^{-N}
\end{equation}

\noindent where equality holds only for pure states.


Wigner considered properties which one would want such a distribution to satisfy and then he showed that the distribution given by equation \ref{wigner} was the only one which satisfied these properties. Some of the properties for a distribution function, $W(\textbf{r},\textbf{p})$, which were considered of special interest are
\begin{itemize}
    \item $W(\textbf{r},\textbf{p})$ should be a Hermitian form of the state vector $\psi(\textbf{r})$, i.e. W is given by
    \begin{equation}
        W(\textbf{r},\textbf{p}) = \bra{\psi} M(\textbf{r},\textbf{p}) \ket{\psi}
    \end{equation}
    where $M(\textbf{r},\textbf{p})$ is a self-adjoint operator. Therefore, $W(\textbf{r},\textbf{p})$ is real.
    \item Besides equations \ref{eq:marginal1} and \ref{eq:marginal2}, it must satisfy the normalization relation mentioned before.
    \item $W(\textbf{r},\textbf{p})$ should be Galilei invariant.
    \item W(\textbf{r},\textbf{p}) should be invariant with respect to space and time reflections.
    \item If $W_{\psi}(\textbf{r},\textbf{p})$ and $W_{\phi}(\textbf{r},\textbf{p})$ are the distributions corresponding to the states $\psi(\textbf{r})$ and $\phi(\textbf{r})$ respectively then
    \begin{equation}
        \left | \int \psi^*(\textbf{r})\phi(\textbf{r}) d^3\textbf{r} \right |^2 = (2\pi\hbar)^3\int d^3\textbf{r}\int W_{\psi}(\textbf{r},\textbf{p})W_{\phi}(\textbf{r},\textbf{p})d^3\textbf{p}
    \end{equation}
    \item Taking into account the Fourier transform $\phi(\textbf{p})$ of the wave function $\psi(\textbf{r})$, equation \ref{wigner} can be re-written in the form
    \begin{equation}
        W(\textbf{r}, \textbf{p}) = \frac{1}{(2\pi)^3}\int \psi^{*}\left(\textbf{p}+\frac{\hbar}{2}\textbf{t}\right)
    \psi\left(\textbf{p}-\frac{\hbar}{2}\textbf{t}\right)
    e^{i\textbf{r}\cdot \textbf{t}} d^3\textbf{t},
    \end{equation}
    exhibiting the basic symmetry under the interchange $q \leftrightarrow  p$.
\end{itemize}
Notice that the very definition of Wigner function given by Eqn. \ref{wigner} relies on a map that takes quantum states as inputs. In what follows, we focus on an approach that is completely independent of the Hilbert space model for quantum theory.

\section{Measure theory and standard probabilities}

Given an outcome set $\Omega$ and a $\sigma-$algebra of subsets of it, $\mathcal{F}$, it is possible to define the probability $p$ as a non-negative real-valued function $p:\mathcal{F}\rightarrow [0,1]$ satisfying the following properties.
\begin{description}
\item[K1.] $p(\Omega)=1$
\item[K2.] For every denumerable and disjoint family $\{A_{i}\}_{i\in\mathbb{N}}$, $p(\bigcup A_{i})=\sum_{i} p(A_{i})$.
\end{description}

\noindent The above equations are called Kolmogorov's axioms. A triplet $(\Omega,\mathcal{F},p)$ is called a \textit{probability space}. An important definition in what follows is that of a \textit{random variable}:

\begin{definition}
\label{def:random-variables}
Let $(\Omega,\mathcal{F},p)$ be a probability space, and let $(\mathbb{R},\mathcal{B})$ be a Borel space with elements of $\mathbb{R}$ being real numbers. A (real-valued) \emph{random variable} $\mathbf{f}$ is a measurable function $\mathbf{f}:\Omega \rightarrow \mathbb{R}$, i.e. for all $B\in \mathcal{B}$, $\mathbf{f}^{-1}(B)\in \mathcal{F}$.
\end{definition}

Random variables express in a technical way the idea of observables in classical probabilistic theories. As an example, consider a classical one-dimensional Harmonic oscillator. The energy -- expressed by the formula $H=\frac{p^{2}}{2m}+\omega^{2}x^{2}$ -- is a function of position and momentum. Each possible value of energy, say, $H = \epsilon_{0}$, can be represented by all possible states in the space $\Gamma = \{(p,q)\,|\,p,q\in\mathbb{R}\}$ that satisfy that condition. Consider the set $H^{-1}([\epsilon_{0},\epsilon_{0}+\delta])=\{(p,q)\,|\,\epsilon_{0}\leq H(p,q)\leq\epsilon_{0}+\delta\}$. Assume that the system is represented by a probabilistic state $p:\Gamma\longrightarrow[0,1]$ (which is a probability density). The probability that the system has energy between $\epsilon_{0}$ and $\epsilon_{0}+\delta$ is then given by $\int_{H^{-1}([\epsilon_{0},\epsilon_{0}+\delta])} p d\lambda $ (where $\lambda$ is the Lebesgue measure in $\Gamma$). Thus, the mathematical concept captured by the notion of random variable is that of a function whose pre-image on any real interval gives place to a measurable set (i.e., a set with a well defined probability).

In order to introduce negative probabilities in physics without appealing to the Hilbert space formalism, one could try by simply extending Kolmogorov's axioms to signed measures in a very direct way:

\begin{definition}
\label{def:signed-measure}
Let $\Omega$ be a sample space and $\mathcal{F}$ a $\sigma$-algebra over $\Omega$. A \emph{signed measure} is a function $\mu : \Sigma\rightarrow \mathbb{R}$ such that
\begin{equation}
    \mu (\emptyset)= 0
\end{equation}
and for every denumerable and disjoint family $\{A_{i}\}_{i\in\mathbb{N}}$
\begin{equation}
    \mu (\bigcup_{i} A_{i})= \sum_{i}\mu (A_{i})
\end{equation}
The triple $(\Omega,\Sigma,\mu)$ is called a \emph{signed measure space} \cite{halmos_measure_1974}.
\end{definition}

But it turns out the above definition of signed measure space is too general for doing physics, since in addition, we need to complement it with a more specific notion of \textit{measurement context}. We address this problem in the following section.

\section{Negative probabilities and measurement contexts}\label{s:Main_Definitions}

In this section we review the signed probabilities introduced in \cite{Indistinguihsability_and_Negative_Probabilities}, with some important modifications. The main features of the signed measures used in this work are:

\begin{itemize}
    \item They are straightforward extension of Kolmogorov's theory and are based in measure theory. They can be used to describe infinite dimensional and non-discrete models as well.
    \item They incorporate the notion of \textit{measurement context} from the very beginning.
    \item They do not rely on the quantum mechanical formalism. They can be computed out of measurement statistics defined in a purely operational way.
\end{itemize}

\subsection{Signed probabilities as signed measure spaces endowed with Kolmogorovian subspaces}\label{s:Properties}

Here we provide a definition of negative probabilities using only measure theoretic notions, which is a simple generalization of Kolmogorov's framework. The key idea is that we start with a signed measure space $(\Omega,\Sigma,\mu)$, which is normalized to unity. Contexts, if they exist, are represented by subspaces $(\Omega,\Sigma_{k},\mu|_{\Sigma_{k}})$ which are Kolmogorovian. The $\Sigma_{k}$'s are taken to be sub $\sigma$-algebras of $\Sigma$. Here $\Sigma$ will be formed by subsets of $\Omega$. Therefore, that applies to the elements of $\Sigma_{k}$ too (for all $k$). In a more general formulation, one could use a more general notion of $\sigma$-algebra (i.e., not based on subsets), but here, we will remain close to the standard one based in subsets of $\Omega$. The intuitive idea behind these choices is that one works with special subsets of $\mathcal{P}(\Omega)$ (for example, the Borel sets), in order to avoid patologic examples (such as the Vitali set).  The first definition that we present in this section goes in the same line as that of Definition 8 in reference \cite{deBarros-Oas-Suppes}. In the reminder of this section, we will give alternative definitions, which illustrate other features of negative probabilities.

In what follows, let $I$ be an arbitrary collection of indexes (not necessarily denumerable). We then define:

\begin{definition}
\label{def:SignedProbaBoolean_Algebras}
Let $(\Omega,\Sigma,\mu)$ be a signed measure space. The triplet $(\Omega,\Sigma,\mu)$ is a \emph{negative probability space}, if it is endowed with a non-empty set of subspaces $(\Omega,\Sigma_{k})$, with $\{\Sigma_{k}\}_{k\in I}\subseteq \Sigma$, such that $\mu(\Omega)=1$ and $(\Omega,\Sigma_{k},\mu|_{\Sigma_{k}})$ is a Kolmogorovian probability space for all $k\in I$.
\end{definition}

Notice that, by construction, $\mu|_{\Sigma_{i}}(E)=\mu|_{\Sigma_{j}}(E)$, whenever $E\in \Sigma_{i}\cap\Sigma_{j}$. This grants that we are working with no-signal models. Notice also that we demand that  $\Sigma_{k}\subseteq\Sigma$, and then, its elements are ultimately elements of $\mathcal{P}(\Omega)$ (given that $\Sigma\subseteq\mathcal{P}(\Omega)$). In general, a physical system will have many physically different states, and then, we give the following definition:

\begin{definition}
Let $(\Omega,\Sigma)$ be a measurable space and let $(\Omega,\Sigma_{k})$ for $k\in I$ a collection of subspaces. A negative probability model will be determined by a non-empty set of signed measures $\mathcal{C}$ such that every $\mu\in \mathcal{C}$ is a negative probability with regard to $(\Omega,\Sigma)$ and $(\Omega,\Sigma_{k})$.
\end{definition}

Notice that, in the above definition, all possible measures $\mu\in\mathcal{C}$ have the same family of contexts. This reflects what happens in many relevant probabilistic theories, such as quantum mechanics and its no-signal generalizations.

As an example, consider the Wigner transformation of a quantum state. The quantum state is now represented by a signed measure in phase space. Since the possible values of $p$ and $q$ are both $\mathbb{R}$, we have that $\Omega=\mathbb{R}\times\mathbb{R}$, and $\Sigma$ can be taken to be the Borel subsets of $\Omega$. Of course, for an arbitrary quantum state $\rho$, and a Borel set $\Delta\in\Sigma$, its phase space representative might take a negative value, i.e. $\mu_{\rho}(\Delta)< 0$. But marginals must be positive, and then, propositions related to $p$ or $q$ only, must have positive values. An event such as ``the value of $q$ lies in the interval $\Delta$" ($\Delta\in\mathcal{B}(\mathbb{R}$) a Borel subset of $\mathbb{R}$), can be reprsented as a subset of $\Omega$ as $\Delta\times \mathbb{R}$. As such, it is a measurable set, and it is obvious that all elements of that form (i.e., $\Delta\times\mathbb{R}$, with $\Delta\in\mathbb{R}=\Omega_{1}$), for a sub $\sigma$-algebra $\Sigma_{q}$ of $\Sigma$. As is well known, $\mu_{\rho}|_{\Sigma_{p}}$ can only take positive values. A similar consideration applies to $\Sigma_{p}$ and $\mu_{\rho}|_{\Sigma_{p}}$. Thus, we see that the wigner quasiprobability distribution Satisfies definition \ref{def:SignedProbaBoolean_Algebras} in a very direct way. The relevance of our reformulation of the definition of negative probabilities, is that we no longer rely on the quantum formalism, allowing for more general probabilistic models. We also show that it suggests a very natural definition of contextuality measure (which can be computed in many examples of interest).

In the following sections we show that, given an arbitrary family of random variables grouped in measurement contexts (described by Kolmogorov spaces $(\Omega_{k},\Sigma_{k},\mu_{k})$), under certain conditions, one can build -- in a canonical way -- a space $(\Omega,\Sigma, \mu)$ in such a way that definition \ref{def:SignedProbaBoolean_Algebras} is satisfied. This point is very important, because it allows to build a connection between actual experimental situations (which can always be ultimately described using collections of random variables) to the measure-theoretic framework described above. Selecting a group of observables and their outcome sets is therefore a natural starting point for many approaches to contextuality (see for example \cite{Abramsky_2011}).

\subsection{Signed probabilities starting with random variables}

In many experimental situations, the notion of random variable can be taken as primitive. Here we provide an alternative definition of negative probabilities that starts with random variables.

\begin{definition}
\label{def:extended-random-variables}
Let $(\Omega,\Sigma,\mu)$ be a signed measure space, and let $(\mathbb{R},\mathcal{B})$ be a Borel space with elements of $\mathbb{R}$ being real numbers, i.e. $\mathcal{B}$ is a $\sigma$-algebra over $\mathbb{R}$. A (real-valued) \emph{extended random variable} $N$ is a measurable function $N:\Omega \rightarrow \mathbb{R}$.
\end{definition}

Since \textit{measurement contexts} are key to our approach we define:

\begin{definition}
\label{def:NEWcontext}
Let $\{N_i\}$, $i=1,\dots, n$, be a collection of extended random variables defined on a signed measure space $(\Omega,\Sigma,\mu)$. A \emph{ $\mu$-induced context} is a subset  $C^{\mu}_j=\{N_{k}\}_{k\in I_j}$, $I_j \subset \{ 1,\ldots, n\}$, for which there exists a sub-$\sigma$-algebra $\Sigma_{j}$ of $\Sigma$ such that, by defining $p^{\mu}_{j}(E):=\mu(E)$ for all $E\in\Sigma_{j}$, the triad $(\Omega,\Sigma_{j},p^{\mu}_{j})$ becomes a probability space, and $N_{i_{k}}$ is a random variable with respect to it, for all $k\in \{1,...,n_{j}\}$.
\end{definition}

Intuitively, a measurement context of a signed measure space, is a collection of -- possibly negative -- random variables for which, the global measure restricted to the Boolean subalgebra associated to those random variables is a Kolmogorovian one.

Given a base Boolean algebra, it is useful to define a family of signed measures over it:

\begin{definition}
Let $\Omega$ be a set and $\Sigma$ a $\sigma$-algebra of subsets of $\Omega$. A \emph{family of signed probabilistic models for $(\Omega,\Sigma)$} is a collection $\mathcal{S}_{(\Omega,\Sigma)}$ of signed measures on $(\Omega,\Sigma)$ such that, for all $\mu\in \mathcal{S}_{(\Omega,\Sigma)}$, $\mu(\Omega)=1$. Any $\mu\in \mathcal{S}_{(\Omega,\Sigma)}$ is called a \emph{state of the model}.
\end{definition}

Using that, we can define a notion of context in connection to a signed family of measures:

\begin{definition}
\label{def:GlobalContext}
Consider a family of signed probability models $\mathcal{S}_{(\Omega,\Sigma)}$. Let $\{N_i\}$, $i=1,\dots, n$, be a collection of  extended random variables defined on $\mathcal{S}_{(\Omega,\Sigma)}$. A \emph{general context} is a subset  $C_j=\{N_{k}\}_{k\in I_j}$, $I_j \subset \{ 1,\ldots, n\}$ of those extended random variables, for which there exists a sub-$\sigma$-algebra $\Sigma_{j}$ of $\Sigma$ satisfying that, for all $\mu\in\mathcal{S}$, by defining $p^{\mu}_{j}(E):=\mu(E)$ for all $E\in\Sigma_{j}$, the triad $(\Omega,\Sigma_{j},p^{\mu}_{j})$ becomes a probability space, and $N_{i_{k}}$ is a random variable with respect to it, for all $k\in \{1,...,n_{j}\}$.
\end{definition}

The above definition makes the notion of context robust with respect to a given family of measures.

Finally, we are now ready for providing a definition of negative probability:

\begin{definition}
\label{def:SignedProba}
A \emph{signed probability space}, also called here \emph{negative probability space}, is a signed measure space $(\Omega,\Sigma,\mu)$ endowed with a non-empty set of contexts $C=\{C^{\mu}_{j}\}$ (in the sense of Definition \ref{def:NEWcontext}), such that $\mu(\Omega)=1$. The measure $\mu$ in this space is a \emph{signed probability} or \emph{negative probability}.
\end{definition}

Notice that we can strengthen the above definition by defining a family of signed probability models, and by making the contexts robust with regard to that family. Also, it could be possible that one is interested in defining a notion of context with regard to a special family of signed measures.

It should be clear that Definition \ref{def:SignedProba} is general enough to cover many relevant examples in physics and statistics. In particular, it is well suited for describing quantum systems and non-signal probabilistic theories in general.

\subsection{Categorical viewpoint}

In this short subsection we review some relevant constructions from measure theory and category theory.
Let us first recall some definitions from measure theory (see \cite[\S 7]{schiling-measure}).
%
%
%

If $(X,\mathcal{A},\mu)$ is a measure space and $f:X\to Y$ is a measurable function,
then the \emph{image measure} (or the \emph{push-forward measure}) is a measure on $(Y,\mathcal{B})$
denoted $f_*(\mu)$ and defined as $f_*(\mu)(B):=\mu(f^{-1}(B))$.
An interesting example of image measure is the \emph{marginals} defined over the product space $X\times Y$.
Recall
that the $\sigma$-algebra on $X\times Y$, often denoted as $\mathcal{A}\otimes\mathcal{B}$,
is the $\sigma$-algebra generated by $\mathcal{A}$ and $\mathcal{B}$. Given a measure $\nu$ on $X\times Y$,
we can define a measure on $X$ by taking the image measure under the projection $\pi_1$.
Then, we define a measure over $X$ as $\pi_{1*}(\nu)(A)=\nu(A\times Y)$.



\

Now, let us give a stronger definition
of signed probability space by using some constructions from category theory.
Assume that we have a family of probability spaces $\mathcal{F}=\{(\Omega_i,\Sigma_i,p_i)\}_{i\in I}$.
A \emph{cone} for the family $\mathcal{F}$ is a measure space $(C,\Sigma,\nu)$ such that there exist measurable maps $c_i:C\to \Omega_i$
with $c_{i*}(\nu)=p_i$ for all $i$. A \emph{universal measure space} for the family $\mathcal{F}$ is a cone $U$ such that any other cone
factorizes through $U$,
\[
\xymatrix{
U\ar[r]^{u_i}&\Omega_i\\
C\ar[ur]_{c_i}\ar@{-->}[u]^{\exists! c}&
}
\]
In general, there exists no universal measure space (since, as we noticed there may be
non-isomorphic measure spaces satisfying the universal property). But, it is true
that there exists a unique (up to unique isomorphism) universal measurable space given by the product $U=\prod_i \Omega_i$ with
the  $\sigma$-algebra $\otimes_i \Sigma_i$.
\begin{definition}
Let $\mathcal{F}=\{(\Omega_i,\Sigma_i,p_i)\}_{i\in I}$ be a
family of probability spaces. A \emph{categorical signed probability space}
is a measure $\mu$ over $\Omega:=\prod \Omega_i$ with marginals $p_i$, that is, $\pi_{i*}(\mu) = p_i$.

If we are interested in measures $\mu$ that satisfies some set of equations $\mathcal{E}$,
we say that $(\prod\Omega_i,\otimes \Sigma_i,\mu)$ is a \emph{categorical signed probability space}
for the family $\mathcal{F}$ with constrains $\mathcal{E}$.
\end{definition}

\

Let us compare this definition with the definition of \emph{signed probability space}. Assume we have a collection
of random variables $N_i^j:\Omega_i\to\mathbb{R}$, where $(\Omega_i,\Sigma_i,p_i)$ is a probability space and
$j\in J$ is some index.
Then, over the product space $\Omega := \prod \Omega_i$, assume there exists a measure $\mu$ such that
$\pi_{i*}(\mu)=p_i$ (we may assume that $\mu$ satisfies some constrains).
With these assumptions, let us construct a signed probability space over $\Omega$.

First, $C_j:=\{N_i^j\circ \pi_i\}$ is a collection of random variables over $\Omega$.
Second, let $\Sigma_i'$ be the $\sigma$-subalgebra of $\Sigma$ defined as $\pi_i^{-1}(\Sigma_i)$.
Notice that any $E'\in \Sigma_i'$ is equal to $\pi_i^{-1}(E)$ for some $E\in\Sigma_i$, then
\[
\mu(E') = \mu(\pi_i^{-1}(E)) = p_i(E).
\]
Hence, $p_i^\mu$ as defined in Definition \ref{def:NEWcontext} is essentially equal to $p_i$.
Then, the  triad $(\Omega,\Sigma_i,p_i^\mu)$ becomes a probability space
and $N_i^j\circ \pi_i$ is a random variable with respect to it. This implies that
the categorical signed probability space
$(\Omega,\Sigma,\mu)$ is a signed probability space with the set of contexts $\{C_j\}$.

Now, is it possible to reverse this construction?
Given a signed probability space, is it possible to give it a structure of categorical signed probability space?
In general, the answer is no. Hence,
the definition of signed probability space given in Definition \ref{def:SignedProba} is more
general than the categorical one.

\section{Properties}\label{s:Properties}

In this section we discuss different mathematical properties of our definition of negative probabilities. We pay special attention to the problems of existence and uniqueness/non-uniqueness. We deal with the problem of finding a signed measurable space $(\Omega,\Sigma, \mu)$ for an arbitrary family of random variables grouped in measurement contexts (described by Kolmogorov spaces $(\Omega_{k},\Sigma_{k},p_{k})$).

Let $J$ and $I_{i}$ be index sets, that could be discrete or continuous/finite or infinite. Here we pose the following problem: given a family of collections of random variables $C_{i}=\{f_{i,j}\}_{j\in I_{i}}$ (for each $i\in J$), with associated Kolmogorov spaces $(\Omega_{i},\Sigma_{i},p_{i})$, under which conditions can we grant the existence of a signed probability -- as in Definition \ref{def:SignedProba} -- for which the $C_{i}$ are measurement contexts (as in Definition \ref{def:NEWcontext})? Before analyzing the general problem, we focus on a very simple example with three random variables.

\subsection{Three dichotomous random variables}

Let us illustrate with a concrete example (taken from \cite{Indistinguihsability_and_Negative_Probabilities}) how the existence problem works and why its solutions are not necessarily unique. Consider three dichotomic random variables $X$, $Y$ and $Z$. This means that they can take two values, say, $1$ and $-1$, and that they have assigned outcome spaces $\Omega_{X}=\{x,\bar{x}\}$, $\Omega_{Y}=\{y,\bar{y}\}$ and $\Omega_{Z}=\{z,\bar{z}\}$ (i.e., $1=x=y=z$ and $-1=\bar{x}=\bar{y}=\bar{z}$). These outcome spaces give place to three different Boolean algebras (given by the power sets of $\Omega_{X}$, $\Omega_{Y}$ and $\Omega_{Y}$: $\mathcal{P}(\Omega_{X})$, $\mathcal{P}(\Omega_{Y})$ and $\mathcal{P}(\Omega_{Z})$). In the context $X;Y$, we have an outcome set given by $\Omega_{X;Y}=\Omega_{X}\times\Omega_{Y}=\{(x,y),(\bar{x},y),(x,\bar{y}),(\bar{x},\bar{y})\}$. Similarly, we have $\Omega_{Y;Z}$ and $\Omega_{X;Z}$. Their $\sigma$-algebras are given by the power sets $\mathcal{P}(\Omega_{X;Y})$, $\mathcal{P}(\Omega_{X;Z})$ and $\mathcal{P}(\Omega_{Y;Z})$, respectively. Similarly, we can define a global outcome set $\Omega_{X;Y;Z}=\Omega_{X}\times\Omega_{Y}\times\Omega_{Z}=\{(x,y,z),(\bar{x},y,z),(x,\bar{y},z),(x,y,\bar{z}),(\bar{x},\bar{y},z),(x,\bar{y},\bar{z}),(\bar{x},y,\bar{z}),(\bar{x},\bar{y},\bar{z})\}$ and its associated $\sigma$-algebra $\mathcal{P}(\Omega_{X;Y;Z})$. Notice that both $\Omega_{X;Y;Z}$ and $\mathcal{P}(\Omega_{X;Y;Z})$ might be experimentally inaccessible, given that, as is well known, there are situations for which there exists no global Kolmogorovian probability assignment that reproduces all marginals and correlations for $X$, $Y$ and $Z$. In what follows, we assume that $X$, $Y$ and $Z$ are pairwise measurable, i.e., that there exist (Kolmogorovian) probability measures $p_{X;Y}$, $p_{Y;Z}$ and $p_{X;Z}$, defined over $\mathcal{P}(\Omega_{X;Y})$, $\mathcal{P}(\Omega_{X;Z})$ and $\mathcal{P}(\Omega_{Y;Z})$, respectively.

Notice that all the algebras $\mathcal{P}(\Omega_{X})$, $\mathcal{P}(\Omega_{Y})$, $\mathcal{P}(\Omega_{Z})$, $\mathcal{P}(\Omega_{X;Y})$, $\mathcal{P}(\Omega_{Y;Z})$ and $\mathcal{P}(\Omega_{X;Z})$, are canonically embedded in $\mathcal{P}(\Omega_{X;Y;Z})$. Let us explain this with some examples. Suppose that we consider the proposition ``$X$ has the value $x$". This is represented in $\mathcal{P}(\Omega_{X})$ by the singleton set $\{x\}$. But there exist also a representatives of the same proposition in $\mathcal{P}(\Omega_{X;Y})$ and $\mathcal{P}(\Omega_{X;Y;Z})$. They are given by $\{(x,y),(x,\bar{y})\}$ and $\{(x,y,z),(x,\bar{y},z),(x,y,\bar{z}),(x,\bar{y},\bar{z})\}$, respectively. Both sets have in common that their elements are formed by all possible values of $Y$ (or $Y$ and $Z$), but the value of $X$ is fixed to be $x$. Similarly, $$\{(\bar{x},y,z),(\bar{x},\bar{y},z),(\bar{x},y,\bar{z}),(\bar{x},\bar{y},\bar{z})\}$$
represents the proposition ``$X$ has value $\bar{x}$" in $\mathcal{P}(\Omega_{X;Y;Z})$, $\{(\bar{x},z),(\bar{x},\bar{z})\}$ represents the proposition ``$X$ has value $\bar{x}$" in $\mathcal{P}(\Omega_{X;Z})$, and so on. The reader can easily find representatives of any proposition about $X$, $Y$ and $Z$ in $\mathcal{P}(\Omega_{X;Y;Z})$. Notice also that any joint proposition about $X$ and $Y$ (or about $X$ and $Z$, or $Y$ and $Z$) has a representative in $\mathcal{P}(\Omega_{X;Y;Z})$. For example, ``$X$ has value $x$ and $Y$ has value $\bar{y}$" is represented by $\{(x,\bar{y},z),(x,\bar{y},\bar{z})\}$. For completeness, the top and bottom elements of $\mathcal{P}(\Omega_{X})$ are represented by the top and buttom elements of $\mathcal{P}(\Omega_{X;Y;Z})$.

But the representatives aren't just copies. They also preserve structure in a natural way. For example, the negation of ``$X$ has value $x$ and $Y$ has value $\bar{y}$", is represented by $\{(x,y,z),(x,y,\bar{z}),(\bar{x},y,z),(\bar{x},y,\bar{z}),(\bar{x},\bar{y},z),(\bar{x},\bar{y},\bar{z})\}$ (which is just the set theoretical complement of  $\{(x,\bar{y},z),(x,\bar{y},\bar{z})\}$). Similarly, the representatives preserve join and meet operations. Thus, we have that  $\mathcal{P}(\Omega_{X})$ and $\mathcal{P}(\Omega_{Y})$ are canonically embedded in $\mathcal{P}(\Omega_{X;Y})$ and $\mathcal{P}(\Omega_{X;Y;z})$ (and the same happens with $\mathcal{P}(\Omega_{y})$, $\mathcal{P}(\Omega_{z})$ and $\mathcal{P}(\Omega_{X;Y})$ with regard to $\mathcal{P}(\Omega_{X;Y;z})$, and so on). Clearly, $\mathcal{P}(\Omega_{X;Y;Z})$ is the minimal Boolean algebra containing all the relevant subalgebras for this example.

Now that we have constructed a global algebra, we have the following problem: search for a global probability assignment $\mu:\mathcal{P}(\Omega_{X;Y;Z})\longrightarrow[0,1]$, such that its marginals are coincident with the input probability distributions and their correlations. In order that the marginals are compatible, we reach the a set of linear equations to be made precise in the following.

The first constrain that we impose is normalization:

\begin{align}\label{e:NormaSigned}
\mu_{xyz}+\mu_{\bar{x}yz}+\mu_{x\bar{y}z}+\mu_{xy\bar{z}}+\mu_{x\bar{y}\bar{z}}+ \mu_{\bar{x}y\bar{z}}+\mu_{\bar{x}\bar{y}z}+\mu_{\bar{x}\bar{y}\bar{z}}=1
\end{align}

Alternatively, we have that $p_{X;Y}$, $p_{X;Z}$ and $p_{Y;Z}$:

\begin{subequations}
\begin{align}\label{e:normX;Y}
p_{X;Y}(xy)+p_{X;Y}(\bar{x}y)+p_{X;Y}(x\bar{y})+p_{X;Y}(\bar{x}\bar{y})=1
\end{align}
\begin{align}\label{e:normX;Z}
p_{X;Z}(xz)+p_{X;Z}(\bar{x}z)+p_{X;Z}(x\bar{z})+p_{X;Z}(\bar{x}\bar{z})=1
\end{align}
\begin{align}\label{e:normY;Z}
p_{Y;Z}(yz)+p_{Y;Z}(\bar{y}z)+p_{Y;Z}(y\bar{z})+p_{Y;Z}(\bar{y}\bar{z})=1
\end{align}
\end{subequations}

Next, we have the constraint imposed by mean values of $X$, $Y$ and $Z$:

\begin{subequations}\label{e:MeanValues}
\begin{align}\label{e:MeanX}
& \mu(xyz)-\mu(\bar{x}yz)+\mu(x\bar{y}z)+\mu(xy\bar{z})+\mu(x\bar{y}\bar{z})-\\\nonumber
& \mu(\bar{x}y\bar{z})-\mu(\bar{x}\bar{y}z)-\mu(\bar{x}\bar{y}\bar{z})=\langle X\rangle
\end{align}
\begin{align}\label{e:MeanY}
& \mu(xyz)+\mu(\bar{x}yz)-\mu(x\bar{y}z)+\mu(xy\bar{z})-\mu(x\bar{y}\bar{z})+\\\nonumber
& \mu(\bar{x}y\bar{z})-\mu(\bar{x}\bar{y}z)-\mu(\bar{x}\bar{y}\bar{z})=\langle Y\rangle
\end{align}
\begin{align}
& \mu(xyz)+\mu(\bar{x}yz)+\mu(x\bar{y}z)-\mu(xy\bar{z})-\mu(x\bar{y}\bar{z})-\\\nonumber
& -\mu(\bar{x}y\bar{z})+\mu(\bar{x}\bar{y}z)-\mu(\bar{x}\bar{y}\bar{z})=\langle Z\rangle
\end{align}
\end{subequations}

\noindent The contexts $X;Y$, $X;Z$ and $Y;Z$ impose the following constraints on $\mu$:

\begin{subequations}\label{e:Contexts}
\begin{align}
& \mu(xyz)-\mu(\bar{x}yz)-\mu(x\bar{y}z)+\mu(xy\bar{z})-\mu(x\bar{y}\bar{z})-\\\nonumber
& \mu(\bar{x}y\bar{z})+\mu(\bar{x}\bar{y}z)+\mu(\bar{x}\bar{y}\bar{z})=\langle XY\rangle
\end{align}
\begin{align}
& \mu(xyz)-\mu(\bar{x}yz)+\mu(x\bar{y}z)-\mu(xy\bar{z})-\mu(x\bar{y}\bar{z})+\\\nonumber
& \mu(\bar{x}y\bar{z})-\mu(\bar{x}\bar{y}z)+\mu(\bar{x}\bar{y}\bar{z})=\langle XZ\rangle
\end{align}
\begin{align}
& \mu(xyz)+\mu(\bar{x}yz)-\mu(x\bar{y}z)-\mu(xy\bar{z})+\mu(x\bar{y}\bar{z});\\\nonumber
& \mu(\bar{x}y\bar{z})-\mu(\bar{x}\bar{y}z)+\mu(\bar{x}\bar{y}\bar{z})=\langle YZ\rangle
\end{align}
\end{subequations}

The above equations can be easily solved, since they are linear. Depending on the input probabilities $p_{X;Y}$, $p_{X;Z}$ and $p_{Y;Z}$, the global solution could be negative, positive (i.e., classical), or not exist at all. If the model satisfies the generalized no-signal condition, a solution will always exist. Notice that we have to determine eight unknown quantities (i.e., the values of $\mu$ on the elements of $\Omega_{X;Y;Z}$), and we have seven equations for mean values and normalization. That yields a subdetermined set of equations. Thus, a key feature of the existence of observables for which there exists no joint probability distribution already appears in this simple model: the set of equations for the global measure is subdetermined in the minimal algebra containing all the contexts. A schematic diagram illustrating the relations among the three random variables and their associated algebras is depicted in Figure \ref{f:Three_dichotomic}.

\begin{figure}
\centering
\includegraphics[width=9.5cm]{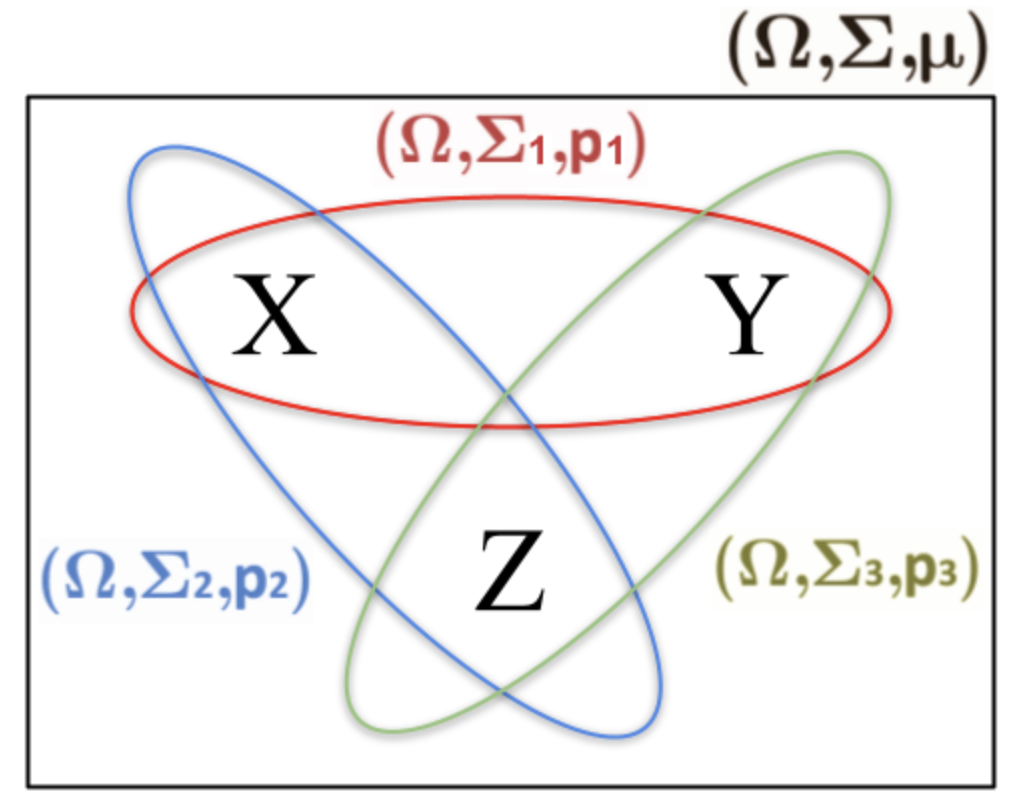}
\caption{The dichotomous random variables, their associated signed measure space and its subspaces. The measurement contexts formed by $XY$, $XZ$ and $YZ$, are illustrated with red, blue and green regions, respectively. The global measure is indicated in black.}
\label{f:Three_dichotomic}
\end{figure}

One can obtain eight equations by fixing the value of the empirically non accessible observable $XYZ$:

\begin{subequations}\label{e:Contexts}
\begin{align}
& \mu(xyz)-\mu(\bar{x}yz)-\mu(x\bar{y}z)-\mu(xy\bar{z})+\mu(x\bar{y}\bar{z})+\\\nonumber
& \mu(\bar{x}y\bar{z})+\mu(\bar{x}\bar{y}z)-\mu(\bar{x}\bar{y}\bar{z})=\langle XYZ\rangle
\end{align}
\end{subequations}

\noindent The above equation represents an experiment that cannot be accessed in our example (i.e., under the assumption that only $XY$, $XZ$ and $YZ$ can be jointly measured). But the input $\langle XYZ\rangle$ can be interpreted as a (non-empirical) parameter that fixes the negative probabilities of the hidden variables.

\subsection{Back to the general case}

More generally, we can pose the following problem. Given a family of pairwise incompatible contexts $C_{1}$, $C_{2}$,...., $C_{M}$, for which joint probability distributions $p_{1}$, $p_{2}$, ...., $p_{M}$ are assumed to exist, we want to know:

\begin{itemize}
  \item (a) The minimal outcome set $\Omega$ and Boolean algebra $\mathcal{B}(\Omega)$ containing the $\sigma$-algebras of each $C_{i}$ and their random variables as subalgebras.
  \item (b) A global (possibly negative) probability assignment $\mu$ satisfying Definition \ref{def:SignedProba}, which is compatible with the $\{p_{i}\}$.
\end{itemize}

To say that the contexts are pairwise incompatible means the following. For each pair of contexts, there will exist a combination of random variables taken from each context, for which the mean value of their product cannot be experimentally realized. In the example of the previous section, if we consider contexts $X-Y$ and $X-Z$, the mean value of $\langle XYZ\rangle$ is not experimentally available (here we took $X$ and $Y$ from context $X-Y$, and $Z$ from context $X-Z$). In a similar way as the example of the previous section, we reach a set of equations for $\mu$. Let us explicitly build the Boolean algebra and formulate the associated equations for a finite collection of random variables with finite outcomes each. Each context $C_{i}$ is built out of $N_{i}$ random variables $f_{ij}$. Given that, in the general formulation, the system might not fulfill the generalized no-signal condition, it is natural to use a notation $f_{ij}$ to describe the $j$-th random variable associated to context $i$. If $f_{ij'}$ and $f_{i'j'}$ (with $i\neq i'$) have the same content but are considered in different contexts, they should not be \textit{a priori} identified \cite{dzhafarov_contextuality_2014}. Each random variable $f_{ij}$, in turn, has associated a family of outcomes $o^{k}_{ij}$, with $k=1,...,\#f_{ij}$, being $\#f_{ij}$ the number of outputs of the random variable $f_{ij}$. Denote by $\Omega_{ij}$ to the outcome set of $f_{ij}$. In what follows, in some situations, we will use a double index $ij$ when we want to make explicit the dependence on the context (the $j$'th random variable of context $i$), and a single index $k$ when we want to refer to a random variable as is uniquely identified by its content (the $k$'st random variable of all possible random variables with a different content).

Let $\mathcal{V}$ be the set of all random variables \textit{having a different content}, and let $V$ be its cardinal. The index $k$, runs from $1$ to $V$: $\mathcal{V}=\{f_{1},f_{2},\cdots,f_{V}\}=\{f_{k}\}_{k=1}^{V}$. Accordingly, denote the outcome set of $f_{k}$ by $\Omega_{k}$. We now proceed to construct a global outcome set $\Omega$ that is formed by considering all possible value specifications for all random variables with a different content. Thus, proceeding similarly as in the example in the previous section, $\Omega$ is formed the Cartesian product of all possible outcomes sets:

\begin{equation}
\Omega = \Omega_{1}\times \Omega_{2}\times\ldots \times \Omega_{V}
\end{equation}

\noindent Each element of $\omega\in\Omega$ is a tuple of the form:

\begin{equation}
\omega = (o^{l_{1}}_{1},o^{l_{2}}_{2},\ldots ,o^{l_{V}}_{V})
\end{equation}

\noindent where the l's run over the number of outputs of each random variable (i.e., $1\leq l_{k}\leq \#f_{k}$). Notice that each $\omega \in \Omega$ is indexed by a list of values for each $f_{k}$. In other words, each $\omega \in \Omega$ specifies a concrete value for each random variable considered.

The minimal Boolean algebra associated to $\Omega$ is $\mathcal{B}=\mathcal{P}(\Omega)$. In what follows, we use a collective variable $w$ to denote each element of $\Omega$. We must now impose several conditions on $\mu$. These are given by the normalization, the mean values of the random variables, and the mean values of all possible $n$-ary products of random variables in all possible contexts, for $2\leq n \leq N_{i}$.

We must specify the conditions of the mean values and correlations. Assume that the context $i$ is formed by the random variables $f_{ij}$, where $1\leq j \leq N_{i}$. Now we have switched again to a notation that makes the dependence on the context explicit. But notice that, as a mathematical object, $f_{ij}$ is equal to an element $f_{k}\in\mathcal{k}$ for some $k$. In what follows, given an $\omega\in \Omega$, let $o_{ij}(\omega)$ be the value taken by $f_{ij}$ in that particular $\omega$. Similarly, $o_{ij}o_{ik}(\omega)$ means the product of the values of $f_{ij}$ and $f_{ik}$ in that $\omega$, and so on. Then, for each possible context $C_i$, and for all indexes $j,k,\ldots$ indexing the random variables, we must have (we include the normalization condition as the first equation below, for completeness and compactness):

\begin{align}\label{eq:General_Conditions_on_Mu}
  \sum_{\omega\in\Omega}\mu(\omega)&=1\\\nonumber
  \sum_{\omega\in\Omega}o_{ij}(\omega)\mu(\omega)&=\langle f_{ij}\rangle\\\nonumber
  \sum_{\omega\in\Omega}o_{ij}o_{ik}(\omega)\mu(\omega)&=\langle f_{ij}f_{ik}\rangle\\\nonumber
  \sum_{\omega\in\Omega}o_{ij}o_{ik}o_{il}(\omega)\mu(\omega)&=\langle f_{ij}f_{ik}f_{il}\rangle\\\nonumber
  &\vdots\\\nonumber
  \sum_{\omega\in\Omega}o_{i1}o_{i2}\ldots o_{iN_{1}}(\omega)\mu(\omega)&=\langle f_{i1}f_{i2}\ldots f_{iN_{i}}\rangle
\end{align}

The above equations are valid for any probabilistic system out of which we can collect its statistics (or at least, of which we can theoretically consider its statistics). The right hand side of equations \ref{eq:General_Conditions_on_Mu} are intended to be computed out of measured data. The left hand side depends on the of $\mu$ that we are looking for, which is determined by the unknown parameters $\{\mu(\omega)\}_{\omega\in\Omega}$. Notice that not all the information might be available: if the system is contextual, only the correlations for some particular subsets of observables will be available (i.e., will give place to realizable measurement contexts). Perhaps, for a particular system, we have, say, only the mean values of binary-products of random variables. Thus, depending on the input information, the solution might not exist, be unique (if a complete set of equations is obtained), or there might be infinitely many solutions.

In case that the model obeys the generalized no-signal condition, the number of equations that can be extracted out of the different contexts \ref{eq:General_Conditions_on_Mu} will be shorter than the number of unknown parameters needed to determine $\mu$. Let us quickly indicate why this is so. Recall that each context $C_{i}$ has $N_{i}$ random variables and that there are $M$ contexts, but a random variable may appear in more than one context. Now, a full identification between random variables with the same content is done, because we are assuming the generalized no-signal condition. Accordingly, we drop again the dependence on the context and use a single index $k$ to denote the random variables (and their associated sets). Since each random variable has $\# \Omega_{k}$ outcomes, we have that the number of elements in $\Omega$ is given by the product of the cardinalities of the outcomes sets of the $f_{k}$'s: $\# \Omega =\# \Omega_{1}\# \Omega_{2}\# \Omega_{3}\cdots \# \Omega_{V}$. The outcome set of a non-trivial random variable has at least two elements, and then $2\leq\# \Omega_{k}$ for all $k$. Thus, we have that $ 2^{V}\leq \# \Omega$. This means that, in order to find $\mu$, the number of unknown parameters is greater or equal than $2^{V}$ (recall that, in order to specify $\mu$, we must determine the parameters $\{\mu(\omega)\}_{\omega\in\Omega}$). Now, let us proceed to determine how many different equations can be extracted from all the possible contexts. If a context $C_{i}$ has $N_{i}$ random variables, it will yield $2^{N_{i}}-1$ different equations. This is so because we need to compute the mean values of the random variables it contains, the mean values of all possible pair products, all possible triplets, and so on. It turns out that there are as many of these mean values as subsets of random variables in context $C_{i}$ (minus one, given that the empty set will give place to no equation). Now, if we consider a new context $C_{j}$ (with $j\neq i$), it will yield $2^{N_{j}}-1$ different equations  again, but some of these equations might be repeated with regard to those of context $C_{i}$. The reason is that $C_{i}$ and $C_{j}$ might share some random variables. Thus, in order to determine an \textit{upper bound} on the number $E_{i,j}$ of different equations can be extracted from $C_{i}$ and $C_{j}$, we must analyze how many equations (for products of random variables) can be generated using $C_{i}\cup C_{j}$. It is crucial to realize that $E_{i,j}$ is strictly shorter than the number of different equations that can be formed using elements from $C_{i}\cup C_{j}$. The reason is that we are assuming that contexts are incompatible, and thus, there will be certain combinations of random variables that cannot be realized together in the same experiment. This means that there will exist at least one combination of random variables, taken from $C_{i}$ and $C_{j}$, for which the mean value of their product will not be empirically available  (in the example of the previous section, the mean value of the product of $X$, $Y$ and $Z$, was not empirically realizable). A similar result holds for three contexts $C_{i}$, $C_{j}$ and $C_{k}$, and so on, until we cover all possible contexts. It turns out that the number $E$ of different equations we can extract from the contexts is strictly shorter than the number of equations we can extract from $C = C_{1}\cup C_{2}\cup\ldots C_{M}$. Since $C = \mathcal{V}$, we have $E < 2^{V}-1$. Since the normalization condition adds a new equation, we conclude that the number $D$ of different equations we can extract from the contexts plus normalization condition satisfies $D = E+1 < 2^{V}$. Thus, the number of different equations is strictly shorter than the number of unknown parameters. Under these conditions, if one solution exists, infinitely many solutions will exist. It is important to remark that, in many cases, it will not be possible to find a positive solution\footnote{The problem of determining the conditions under which a \textit{non-negative} measure exists for a given family of random variables is a rather complicated subject. See for example \cite{Vorobev-1959}.}. In many cases of interest (as in quantum and quantum-like contextuality scenarios), we will find infinitely many (possibly signed) solutions.

It is instructive to revive the example of the previous section under the light of the above proof. In that example, each context has two random variables, and there are three random variables in total (with two outcomes each). The number of unknown parameters is given by $\#\Omega = 2^{3}=8$. Each context gives place to three different equations. For example, context $X-Y$ gives place to the mean values $\langle X \rangle$, $\langle Y \rangle$ and $\langle XY \rangle$. Context $X-Z$ gives place to the mean values $\langle X \rangle$, $\langle Z \rangle$ and $\langle XZ \rangle$. Thus, the mean value $\langle X \rangle$ is repeated. If we sum all the different equations from the three contexts, we obtain six equations ($\langle X \rangle$, $\langle Y \rangle$, $\langle Z \rangle$, $\langle XY \rangle$, $\langle XZ \rangle$ and $\langle YZ \rangle$), and we must add to them the normalization condition: seven equations in total. These are all compatible equations. The number of all possible equations (without taking into account incompatibility) is eight, since we are including the mean value $\langle XYZ \rangle$. Thus, we see that the number of different equations that can be extract from the contexts is strictly shorter than the number of all conceivable equations (i.e., disregarding the incompatibility condition). And the latter is always shorter or equal than the number of unknown parameters.

For the particular case of quantum systems, the mean values in the right hand side of equations \ref{eq:General_Conditions_on_Mu} can be expressed using the Born rule. Furthermore, all quantum observables have the same number of outputs. Denote the observable $j$ of context $i$ by $A_{ij}$, and its outcomes by $o_{ij}$. Thus, for a quantum system prepared in state $\rho$ we have:

\begin{align}\label{eq:General_Conditions_on_Mu_Quantum}
  \sum_{\omega\in\Omega}\mu(\omega)&=\mbox{tr}(\rho)\\\nonumber
  \sum_{\omega\in\Omega}o_{ij}(\omega)\mu(\omega)&=\mbox{tr}(\rho A_{ij})\\\nonumber
  \sum_{\omega\in\Omega}o_{ij}o_{ik}(\omega)\mu(\omega)&=\mbox{tr}(\rho A_{ij}A_{ik})\\\nonumber
  \sum_{\omega\in\Omega}o_{ij}o_{ik}o_{il}(\omega)\mu(\omega)&=\mbox{tr}(\rho A_{ij}A_{ik}A_{il})\\\nonumber
  &\vdots\\\nonumber
  \sum_{\omega\in\Omega}o_{i1}o_{i2}\ldots o_{iN_{1}}(\omega)\mu(\omega)&=\mbox{tr}(\rho A_{i1}A_{i2}\ldots A_{iN_{i}})
\end{align}

In the right hand side of equations \ref{eq:General_Conditions_on_Mu_Quantum}, only compatible observables (contained in a particular measurement context $C_{i}$) are considered. These can represent different parties, or refer to a single quantum system.

\subsection{Selecting a signed probability using the $L_{1}$-norm}\label{s:Uniqueness}

Given a finite dimensional quantum system (for example, a system of qubits in a quantum information devise), prepared in a definite state represented by a density operator $\rho$, there is a definite number of independent measurements that is enough for determining $\rho$ uniquely. For example, in a system of $N$ qubits, an arbitrary density operator is determined by performing at most $2^{2N}-1$ independent measurement statistics\footnote{This number can be optimized. See for example \cite{Holik_2023_Parametrizing}}. The statistics of all other measurement contexts are --so to say-- determined by those values, since they determine the density operator $\rho$ representing the physical state. But notice that, when representing the quantum state using signed measures, if we fix the number of available contexts, its associated minimal Boolean algebra will determine a number of unknown parameters which will be always less than the number of equations empirically available. It is not possible to solve this by adding measurement contexts, given that the number of unknown parameters will be increased. Therefore, we need to face a situation in which there will exist more than one signed measure which is compatible with the observed data. Which one should we choose? We must provide a rule for making a choice.




In what follows, it is important to build a geometric picture of the set of signed measures associated to a number of experimental mean values. Let $(X,\Sigma)$ be a measurable space (representing the global algebra assigned to a given collection of measurement contexts) and let $\mathcal{M}(X)$ be the Banach space
of signed measures over $X$, with finite total variation,
\[
\mathcal{M}(X) = \{\mu\,\colon\,\|\mu\|<\infty\},
\]
where $\|\mu\|=\sup\{\mu(A)-\mu(B)\,\colon\, A,B\in\Sigma\}$.

In $\mathcal{M}(X)$, let us call $L$ to the affine linear space given by a finite
number of linear equations. An example of such a linear equation is the condition $\mu(X)=1$ (notice that, due to normalization, this equation is always present in our approach). More generally, each mean value equation that we can consider in any context, will be represented by linear equations of that form. Thus, $L$ represents the collection of all signed measures which are compatible with our empirical model. As we explained above, the number of linear equations will be, in general, lower than that of unknown parameters. 

How many elements are in $L$? For an arbitrary empirical model, there will be more than one. Only when all possible mean values are available we have a unique solution, but the existence of observables which are not jointly measurable --as in quantum theory-- blocks this possibility. Thus, we are faced with the question: out of all possible elements in $L$, which one should we chose as representative of the empirical system under consideration? There are several options. Among them, one might consider the maximization of entropy or the optimization of any other quantity which has (a) nice mathematical properties and (b) is suitable for describing the physics of the problem. In what follows, we study some of these options. To do that, it is crucial to characterize the geometrical properties of $L$ with some detail, in connection with the quantity to be optimized.

Let us consider first what happens if we try to minimize $\|\mu\|$ (this is the strategy followed in \cite{Indistinguihsability_and_Negative_Probabilities}).
Given that $0\not\in L$ (because $\mu(X)=1$), for $\epsilon>0$ small enough, the convex set $C_\epsilon=\{\mu\in\mathcal{M}(X)\,\colon\,\|\mu\|=\epsilon\}$ is disjoint
from $L$ ($C_\epsilon$ is the ball of radius $\epsilon$). But, if $L\neq \emptyset$ and we chose $\epsilon>0$ big enough, then
$C_\epsilon\cap L\neq\emptyset$. Thus, let $\epsilon_0$ be defined as
\[
\epsilon_0:=\inf \{\epsilon>0\,\colon\,C_\epsilon\cap L\neq\emptyset\}
= \sup
\{\epsilon>0\,\colon\,C_\epsilon\cap L=\emptyset\}.
\]
In general, it can be proved that $C_{\epsilon}\cap L$ is convex. But it is not true that there exits some $\epsilon>0$
such that $C_{\epsilon}\cap L$ is a point. In other words, it is not true that
the norm attains a minimum value
over $L$ (see \cite[\S 5, Exercise 5]{rudin-measure}).

If $X=\mathbb{R}^n$ we can restrict $\mathcal{M}(\mathbb{R}^{n})$
to signed measures which are
absolutely continuous with respect to the Lebesgue measure $\lambda$,
\[
\mathcal{M}^{ac}(\mathbb{R}^{n}) = \{\mu\in\mathcal{M}(X)\,\colon\,\|\mu\|<\infty,\,\mu\ll\lambda\}.
\]
From Radon-Nikodym's theorem we have the following isometric isomorphism,
\[
\mathcal{M}^{ac}(\mathbb{R}^{n}) \cong L^1(\mathbb{R}^n).
\]
where the (inverse) map sends $f$ to the measure $\mu$ defined as
$\mu(A)=\int_A f d\lambda$.
Analogously, if $X$ is a discrete space, we have the following
isometric isomorphism,
\[
\mathcal{M}(X) \cong L^1(X).
\]
In this case, the condition of absolute continuity with respect to the counting measure is vacuous.
The (inverse) map sends the weights $(p_x)_{x\in X}$ to the measure $\mu$ defined as
$\mu(A)=\sum_{x\in A} p_x$.

The above properties of $\|\mu\|$ indicate that it is a reasonable candidate quantity out of which one can build a contextuality measure (see also the discussion in \cite{Indistinguihsability_and_Negative_Probabilities}). Thus, we give the following:

\begin{definition}
Let $X$ be a measurable space and let $\mu\in \mathcal{M}(X)$ be any element of $L$ of minimal norm (i.e., an element of norm $\epsilon_{0}$). Thus \emph{contextuality} of the empirical model is defined as $1-\|\mu\|$.
\end{definition}
Notice that in case $\mu$ is given by $\mu(A)=\int_A f d\lambda$ for some $f\in L^1(X)$,
the contextuality of $\mu$ is equal to $1-\|f\|_1$. Also, if $X$ is discrete,
then $\mu$ is determined by some weights $(p_x)_{x\in X}$ and
the contextuality of $\mu$ is equal to $1-\sum_{x\in X}|p_x|$ (compare with the construction presented in \cite{Indistinguihsability_and_Negative_Probabilities}). The following theorem is useful for our purposes in the rest of this work:

\begin{theorem}
If $X$ is finite and $L$ is transversal to the ball in $L^1(X)$, then there exists a unique signed measure
minimizing the contextuality value.
\end{theorem}
\begin{proof}
If $X$ is finite of cardinality $n$, $L^1(X)$ is isomorphic to $(\mathbb{R}^n,\|\cdot\|_1)$.
Hence, from the hypothesis on $L$,
there exists a unique vector $(p_x)_{x\in X}$ such that $\sum_{x\in X} |p_x|$ is minimum.
\end{proof}

\subsection{Entropic measures}

Usually, in quantum mechanics we have infinitely many different measuremnt contexts. Von Neumann's entropy can be defined as:

\begin{equation}
    S(\rho)= \mbox{tr}(\rho\ln{\rho})
\end{equation}

An equivalent definition is as follows. Given an orthonormal basis $B = \{|v_{i}\rangle\}_{i=1,..,n}$ of the Hilbert space, the entropy relative to that basis is given by:

\begin{equation}
    S_{B}(\rho)= \sum p_{i}\ln(p_{i})
\end{equation}

\noindent where $p_{i}=\mbox{tr}(\rho |v_{i}\rangle\langle v_{i}|)$. It can be proved that the von Neumann entropy satisfies:

\begin{equation}
    S(\rho)= \min_{B}(S_{B}(\rho))
\end{equation}

Here, we can make a similar move, and define the entropy associated to a negative probability as the infimum taken among all the Shannon entropies associated to the considered measurement context. Thus, assume that we are considering the family of contexts $\mathcal{B}=\{B_{i}\}_{i\in I}$. Thus, the entropy associated to the negative probability $\mu$ will be given by the formula:

\begin{equation}
    S_{\mathcal{B}}(\mu)= \min_{B\in\mathcal{B}}(S(\mu|_{B}))
\end{equation}

\noindent Notice that, for a quantum systems, if $\mathcal{B}$ is taken to be all possible measurement contexts (or if it includes the context that diagonalizes the density operator), the above definition coincides with the von Neumann entropy.

\section{Some Applications}\label{s:Applications}

There has been a growing interest in quantum contextuality, due to its possible connection with the performance of quantum computers. For that reason, quantifying quantum contextuality becomes of the essence. Several measures of contextuality has been developed for that aim (see for example \cite{Dzhafarov-Brief_Overview,Abramsky-Barbosa-Mansfield-2017,deBarros-Oas-Suppes,DEBARROS_2016}). Some of them have been compared, yielding similar results in several important examples \cite{deBarros-Dzhafarov-Kujala-Oas-2016}. Here, we focus in the $L_{1}$-norm already discussed in Section \ref{s:Uniqueness}. The reason is that it fits naturally with the negative probabilities approach that we are discussing here. Furthermore, it possesses the advantage of being easily implemented with a Python code. In what follows, we will analyze different quantum contextuality scenarios, and see how the $L_{1}$-norm behaves in them.



\subsection{Entanglement and contextuality scenarios}

In this section we analyze some examples of quantum states and contextuality scenarios. Notice that the computed value of the contextuality measure will depend, in general, of the chosen scenario. In particular, if we choose a system of two qubits and we consider a Bell-type setting, the contextuality will depend on the chosen angles for the observables. Therefore, for a given state, it is reasonable to choose those angles corresponding to its maximal violation of the CHSH inequality.

\noindent\textbf{Cat-like states.}

For two qubits, it is instructive to analyze states of the form:

\begin{equation}
|\psi\rangle =\sqrt{p}|00\rangle+\sqrt{1-p}|11\rangle
\end{equation}

\noindent with $p\in[0,1/2]$ (these are called \textit{Cat-like states}). In Figure \ref{fig:Context_Vio_Ent}, we show the entanglement entropy, contextuality and degree of violation for each state of this family (as $p$ ranges from $0$ to $1/2$). A histogram of the contextuality values obtained for this family is displayed in figure \ref{fig:Histogram_cat_like_twoqubits}. The contextuality is quantified as follows: for each state, we compute the angles corresponding to the maximal value of violation of the CHSH inequality. For those angles, we compute the minimal value of the $L_1$-norm associated to the mean values of the observables constructed with those angles. The maximization with regard to the angles is very important, because a given state might be contextual with regards to some observables, but non-contextual with regards to others. In figure \ref{fig:Comparison} we show a comparison between the procedure taking maximal angles (figure \ref{fig:Comparison} right) vs the procedure without maximization (figure \ref{fig:Comparison} right).

\begin{figure}[H]
\centering
\includegraphics[width=9.5cm]{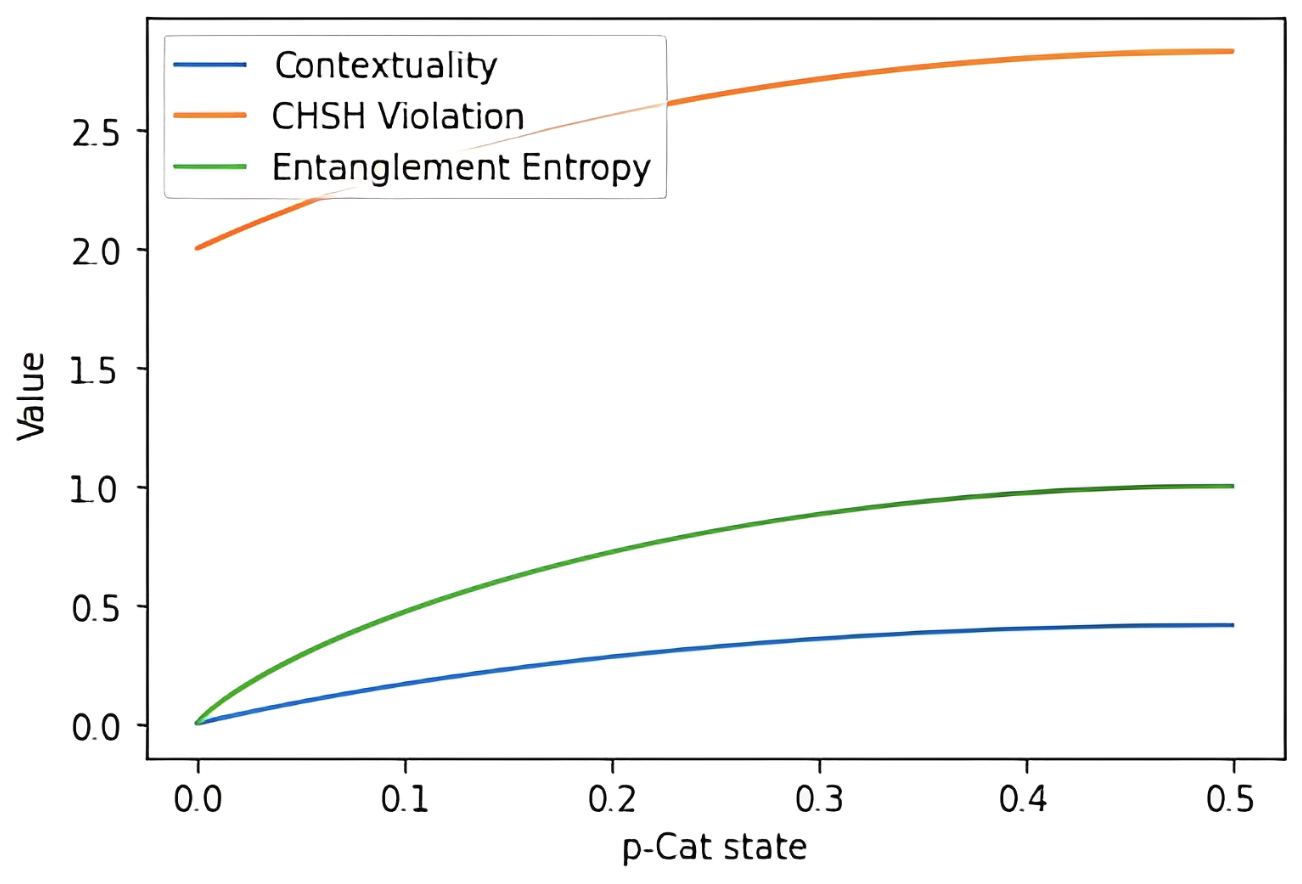}
\caption{Plots of the values of entanglement entropy, contextuality and degree of violation for the cat-like states family, for the case of two qubits.}
\label{fig:Context_Vio_Ent}
\end{figure}

\begin{figure}[H]
\centering
\includegraphics[width=9cm]{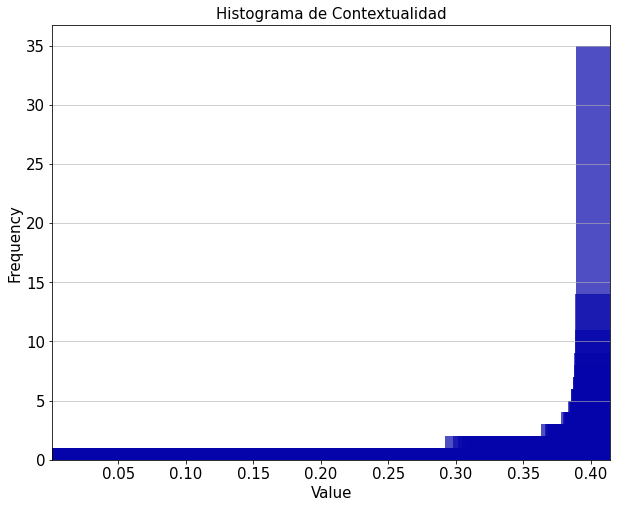}
\caption{Histogram of the contextuality values of the cat-like family of two qubits for $1000$ incraesing values of $p$.}
\label{fig:Histogram_cat_like_twoqubits}
\end{figure}

\begin{figure}[H]
     \centering
     \begin{subfigure}{0.45\textwidth}
         \centering
        \includegraphics[width=\textwidth]{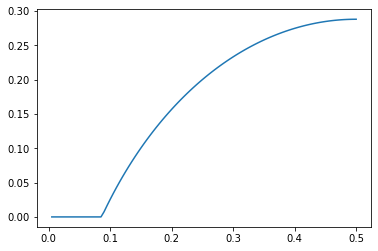}
         \caption{Angles fixed.}
     \end{subfigure}
     \hfill
     \begin{subfigure}{0.45\textwidth}
         \centering
         \includegraphics[width=\textwidth]{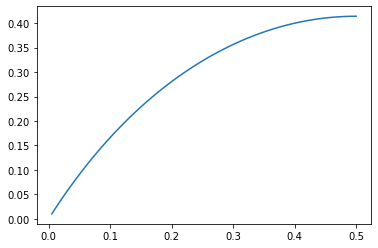}
         \caption{Maximizing angles.}
     \end{subfigure}
              \caption{Maximizing vs not maximizing angles in the CHSH inequality.}
        \label{fig:Comparison}
\end{figure}

\noindent \textbf{Bell-type scenario}

A set of values for the probabilities for the correlations of a Bell-type scenario are displayed in Table \ref{table:Bell} (see \cite{Abramsky_2011}, section $2.6$). These give place to a set of linear equations that can be solved. The minimum $L_{1}$-norm state is taken and this is used to compute the contextuality. For this scenario, the contextuality value obtained is $C = 0.24999999999999245$.

\begin{table}[h]
\centering
\begin{tabular}{|c|c|c|c|c|c|}
\hline
\multicolumn{2}{|c|}{\multirow{2}{*}{\backslashbox{\textcolor{red}{$\mathbf{A}$}}{\textcolor{blue}{$\mathbf{B}$}}}}
 & \multicolumn{2}{c|}{\textcolor{blue}{$\mathbf{a'}$}} &
    \multicolumn{2}{c|}{\textcolor{blue}{$\mathbf{b'}$}} \\
\cline{3-6}
 \multicolumn{2}{|c|}{} & \textcolor{blue}{$+$} & \textcolor{blue}{$-$} & \textcolor{blue}{$+$} & \textcolor{blue}{$-$}\\
\hline
\multirow{2}{*}{\textcolor{red}{$\mathbf{a}$}} & \textcolor{red}{$+$} & $\frac{1}{2}$ & $0$ & $\frac{3}{8}$ & $\frac{1}{8}$\\
\cline{2-6}
 & \textcolor{red}{$-$} & $0$ & $\frac{1}{2}$ & $\frac{1}{8}$ & $\frac{3}{8}$\\
\hline
\multirow{2}{*}{\textcolor{red}{$\mathbf{b}$}} & \textcolor{red}{$+$} & $\frac{3}{8}$ & $\frac{1}{8}$ & $\frac{1}{8}$ &$\frac{3}{8}$ \\
\cline{2-6}
 &\textcolor{red}{$-$} & $\frac{1}{8}$ & $\frac{3}{8}$ & $\frac{3}{8}$ & $\frac{1}{8}$\\
\hline
\end{tabular}
\caption{Bell inequality probabilities.}
\label{table:Bell}
\end{table}




\noindent \textbf{Popescu-Rohrlich box}

The probabilities for the PR box \cite{Popescu1994} are displayed in Table \ref{table:PR}. For this scenario, the contextuality value obtained is maximal (and lies well beyond the quantum limit):  $C=0.9999999999997877$.

\begin{table}[H]
\centering
\begin{tabular}{|c|c|c|c|c|c|}
\hline
\multicolumn{2}{|c|}{\multirow{2}{*}{\backslashbox{\textcolor{red}{$\mathbf{A}$}}{\textcolor{blue}{$\mathbf{B}$}}}}
 & \multicolumn{2}{c|}{\textcolor{blue}{$\mathbf{a'}$}} &
    \multicolumn{2}{c|}{\textcolor{blue}{$\mathbf{b'}$}} \\
\cline{3-6}
 \multicolumn{2}{|c|}{} & \textcolor{blue}{$+$} & \textcolor{blue}{$-$} & \textcolor{blue}{$+$} & \textcolor{blue}{$-$}\\
\hline
\multirow{2}{*}{\textcolor{red}{$\mathbf{a}$}} & \textcolor{red}{$+$} & $\frac{1}{2}$ & $0$ & $\frac{1}{2}$ & $0$\\
\cline{2-6}
 & \textcolor{red}{$-$} & $0$ & $\frac{1}{2}$ & $0$ & $\frac{1}{2}$\\
\hline
\multirow{2}{*}{\textcolor{red}{$\mathbf{b}$}} & \textcolor{red}{$+$} & $\frac{1}{2}$ & $0$ & $0$ &$\frac{1}{2}$ \\
\cline{2-6}
 &\textcolor{red}{$-$} & $0$ & $\frac{1}{2}$ & $\frac{1}{2}$ & $0$\\
\hline
\end{tabular}
\caption{PR box probabilities.}
\label{table:PR}
\end{table}

\noindent \textbf{Mermin correlations}

The probabilities for the Mermin square are displayed in Table \ref{table:Mermin} (see \cite{Three_Mikes_Manifesto}, section $2.4$). For this scenario, the contextuality value obtained is  $C=0.25000000035452974$.

\begin{table}[h]
\centering
\begin{tabular}{|c|c|c|c|c|c|c|c|}
\hline
\multicolumn{2}{|c|}{\multirow{2}{*}{\backslashbox{\textcolor{red}{$\mathbf{A}$}}{\textcolor{blue}{$\mathbf{B}$}}}}
 & \multicolumn{2}{c|}{\textcolor{blue}{$\mathbf{a}'$}} &
    \multicolumn{2}{c|}{\textcolor{blue}{$\mathbf{b}'$} } & \multicolumn{2}{c|}{\textcolor{blue}{$\mathbf{c}'$}} \\
\cline{3-8}
 \multicolumn{2}{|c|}{} & \textcolor{blue}{$+$} & \textcolor{blue}{$-$} & \textcolor{blue}{$+$} & \textcolor{blue}{$-$} & \textcolor{blue}{$+$} & \textcolor{blue}{$-$}\\
\hline
\multirow{2}{*}{\textcolor{red}{$\mathbf{a}$}} & \textcolor{red}{$+$} & $0$ & $\frac{1}{2}$ & $\frac{3}{8}$ & $\frac{1}{8}$ & $\frac{3}{8}$ & $\frac{1}{8}$\\
\cline{2-8}
 & \textcolor{red}{$-$} & $\frac{1}{2}$ & $0$ & $\frac{1}{8}$ & $\frac{3}{8}$ & $\frac{1}{8}$ & $\frac{3}{8}$\\
\hline
\multirow{2}{*}{\textcolor{red}{$\mathbf{b}$}} & \textcolor{red}{$+$} &$\frac{3}{8}$ & $\frac{1}{8}$ & $0$ &$\frac{1}{2}$ & $\frac{3}{8}$ &$\frac{1}{8}$ \\
\cline{2-8}
 &\textcolor{red}{$-$} & $\frac{1}{8}$ & $\frac{3}{8}$ & $\frac{1}{2}$ & $0$& $\frac{1}{8}$ & $\frac{3}{8}$\\
\hline
\multirow{2}{*}{\textcolor{red}{$\mathbf{c}$}} & \textcolor{red}{$+$} &$\frac{3}{8}$ & $\frac{1}{8}$ & $\frac{3}{8}$ &$\frac{1}{8}$ & $0$ &$\frac{1}{2}$ \\
\cline{2-8}
 & \textcolor{red}{$-$} & $\frac{1}{8}$ & $\frac{3}{8}$ & $\frac{1}{8}$ & $\frac{3}{8}$ & $\frac{1}{2}$ & $0$\\
\hline
\end{tabular}
\caption{Mermin square probabilities.}
\label{table:Mermin}
\end{table}










\subsection{Quantum Random Circuits}

\begin{figure}[h!]
     \centering
     \begin{subfigure}{0.45\textwidth}
         \centering
         \includegraphics[width=\textwidth]{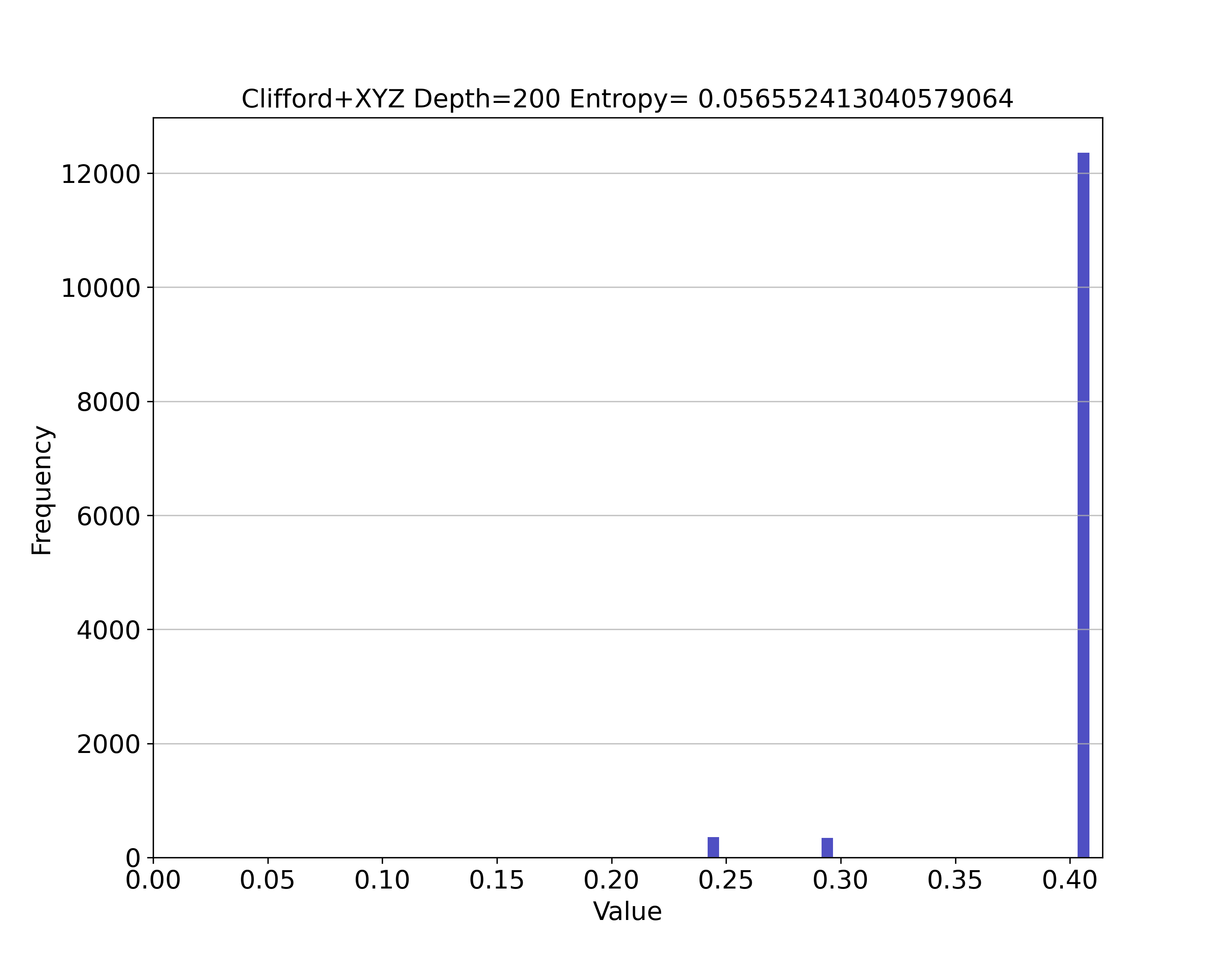}
         \caption{Clifford. Shannon entropy $=$ $0.054$.}
         \label{fig:Comparison_left}
     \end{subfigure}
     \hfill
     \begin{subfigure}{0.45\textwidth}
     \centering
         \includegraphics[width=\textwidth]{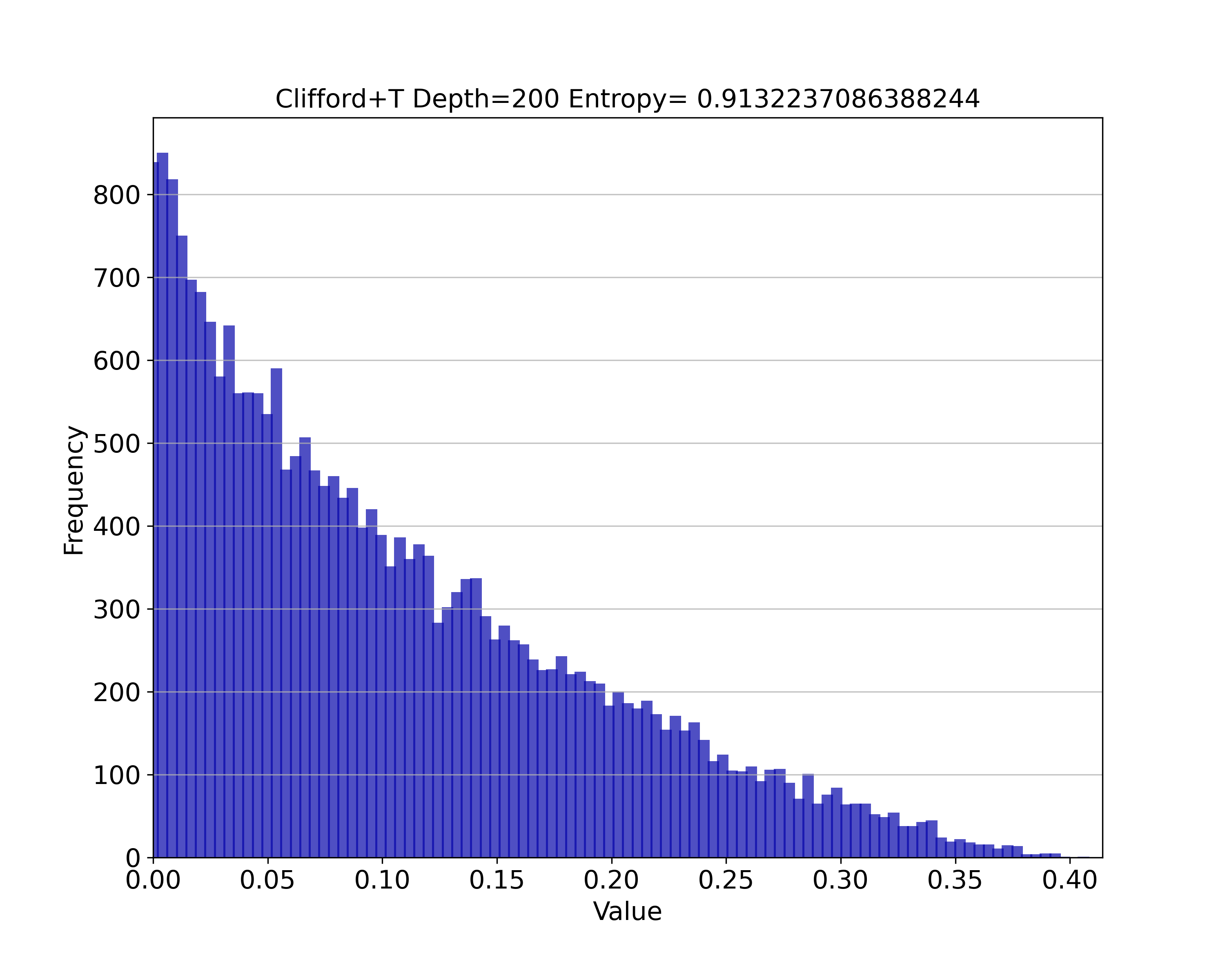}
         \caption{Clifford $+$ $T$. Shannon entropy $=$ $0.91$.}
         \label{fig:Comparison_right}
     \end{subfigure}
     \hfill
     \centering
     \begin{subfigure}{0.45\textwidth}
         \centering
         \includegraphics[width=\textwidth]{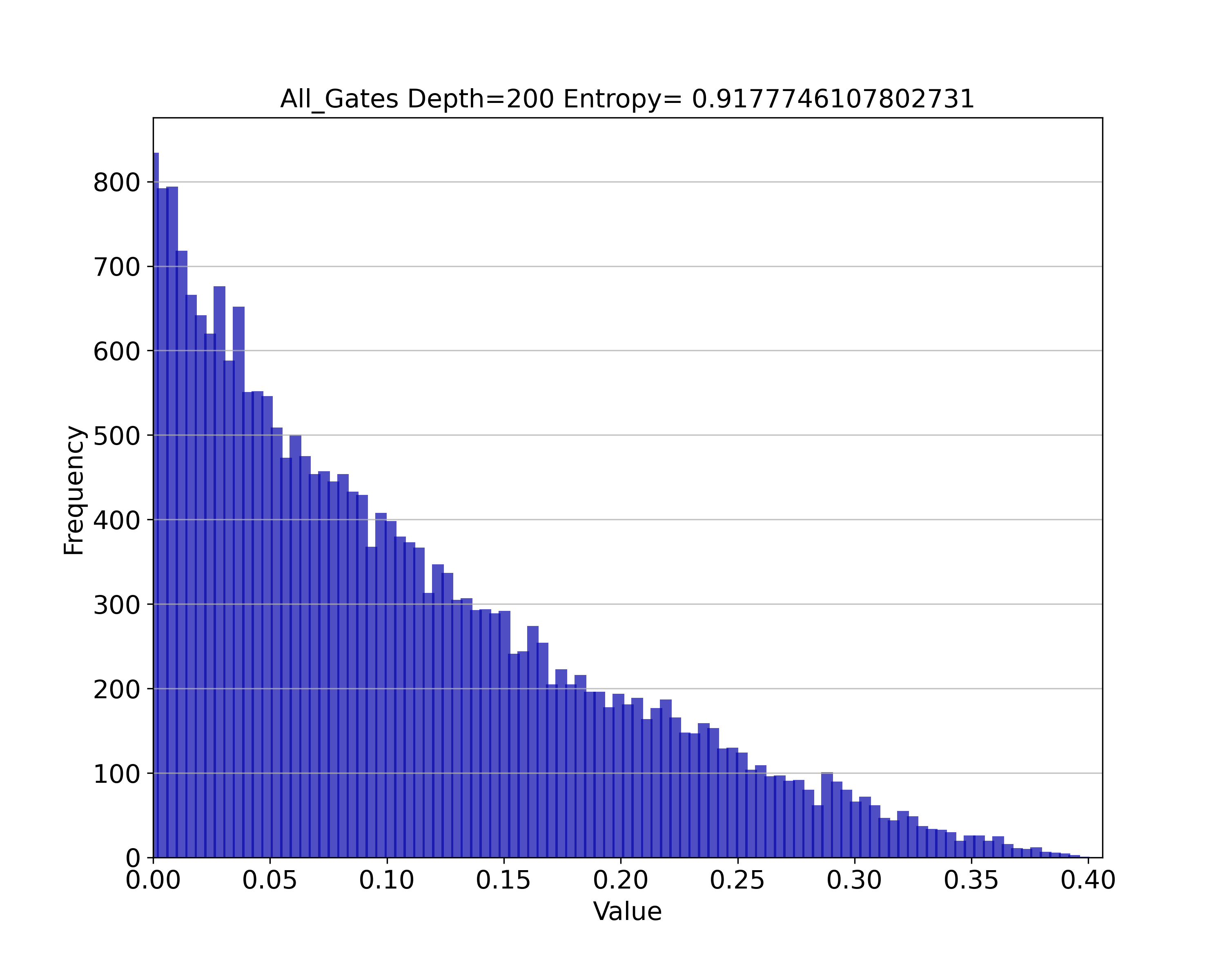}
         \caption{All possible Qiskit gates. Shannon entropy $=$ $0.91$.}
         \label{fig:Comparison_right}
     \end{subfigure}
     \hfill
     \centering
     \begin{subfigure}{0.45\textwidth}
         \centering
         \includegraphics[width=\textwidth]{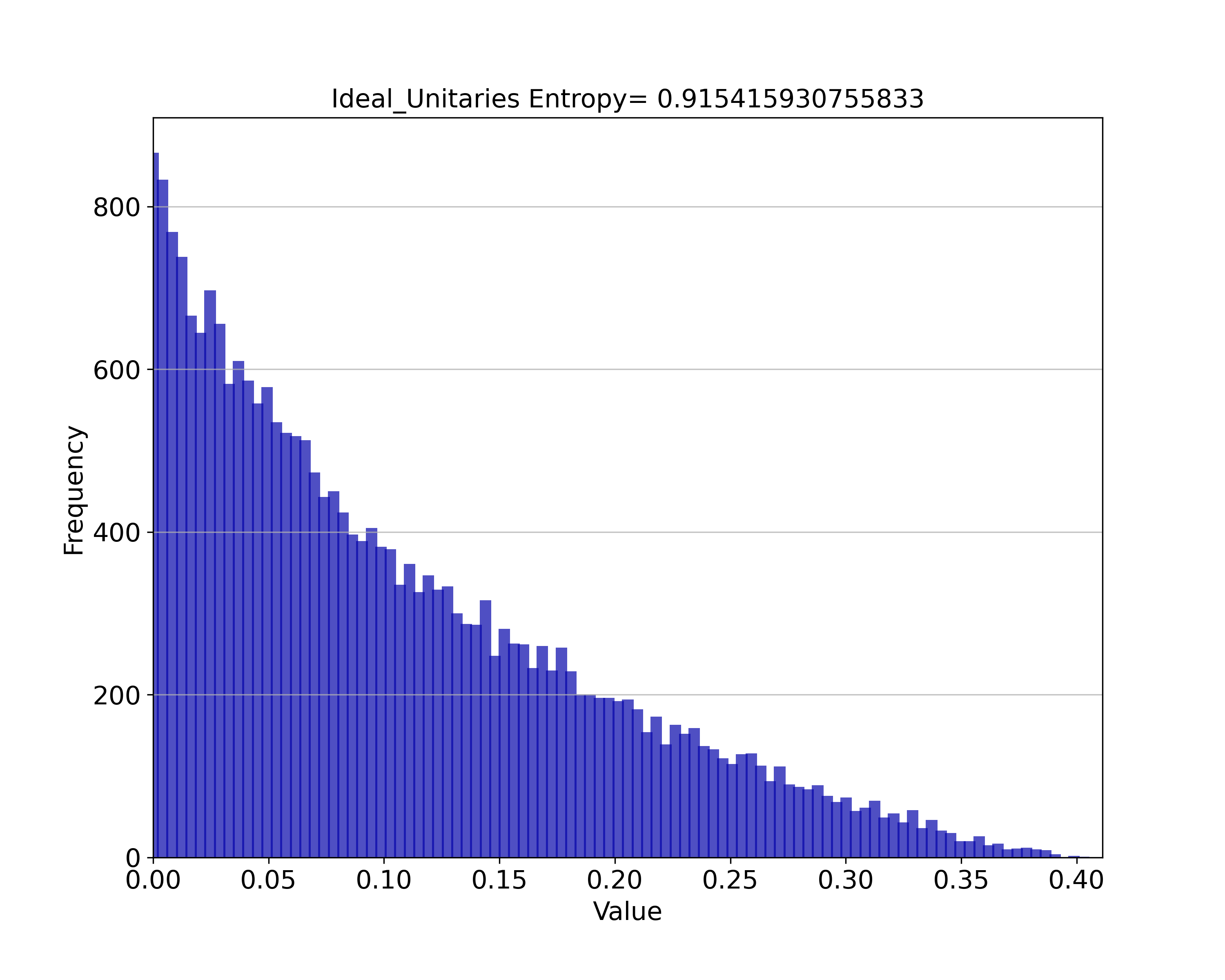}
         \caption{100000 unitary gates generated with the Haar measure. Shannon entropy $=$ $0.91$.}
         \label{fig:Comparison_right}
     \end{subfigure}
     \hfill
              \caption{Probability distributions associated to  the contextuality values for quantum random circuits generated with different sets of elementary gates. Each of the $100.000$ generated circuits has depth$=200$. The zero contextuality pick is not shown, for a better visualization. The probabilities are renormalized accordingly.}
        \label{fig:Comparison_Qiskit}
\end{figure}

\begin{figure}[h!]
     \centering
     \begin{subfigure}{0.45\textwidth}
         \centering
         \includegraphics[width=\textwidth]{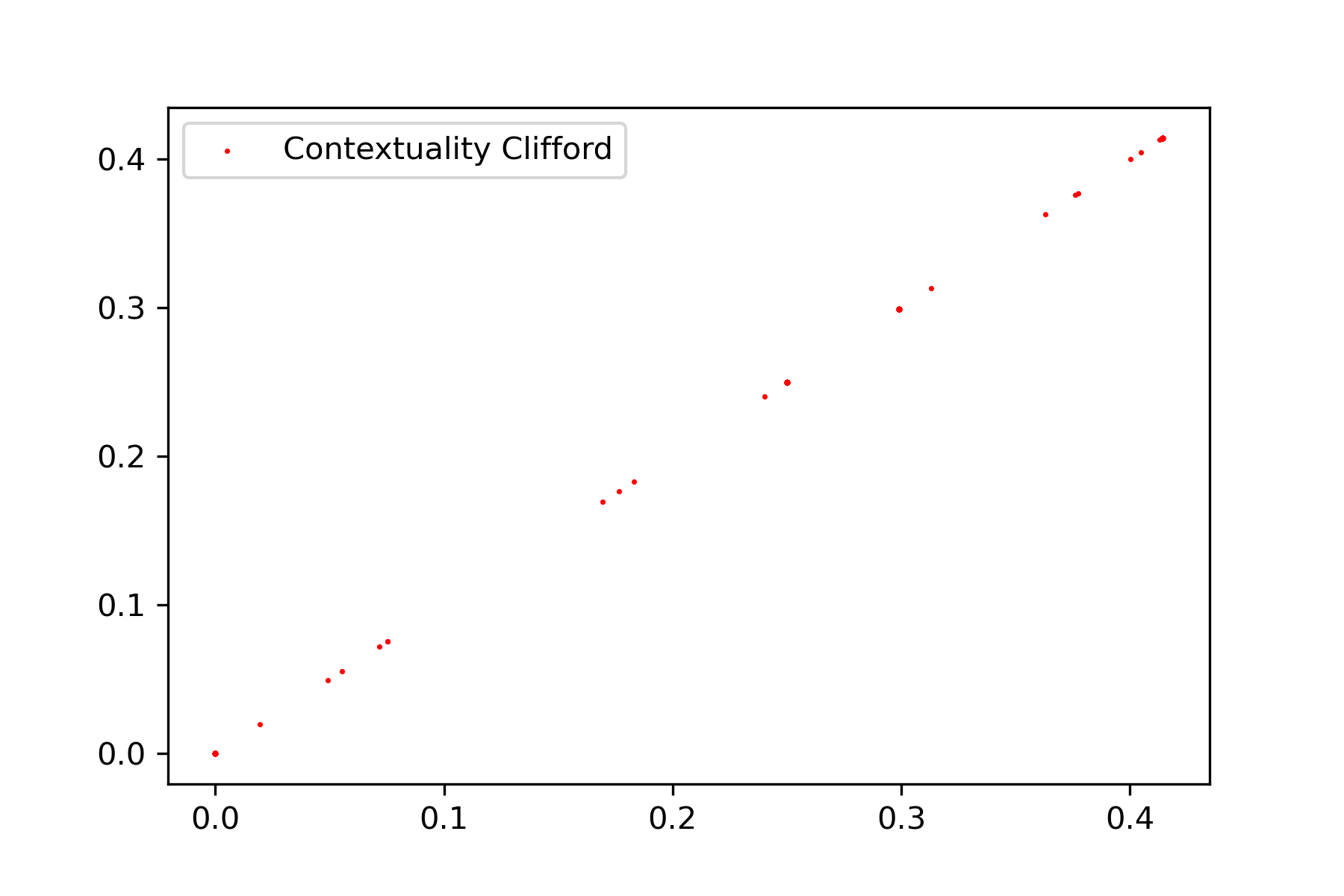}
         \caption{Contextuality values for the Clifford set.}
     \end{subfigure}
     \hfill
     \begin{subfigure}{0.45\textwidth}
     \centering
         \includegraphics[width=\textwidth]{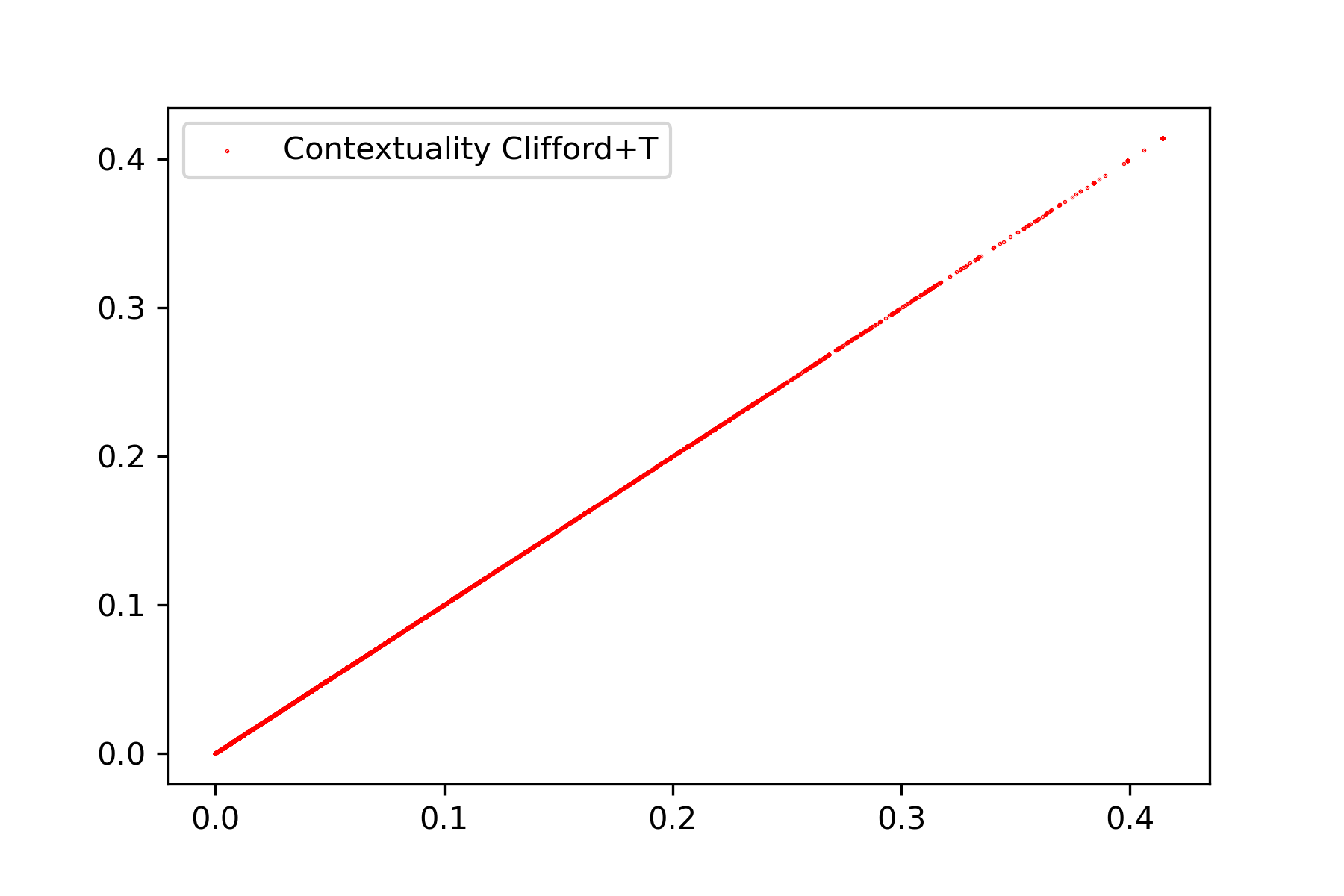}
         \caption{Contextuality values for Clifford $+$ $T$}
     \end{subfigure}
     \hfill
     \centering
     \begin{subfigure}{0.45\textwidth}
         \centering
         \includegraphics[width=\textwidth]{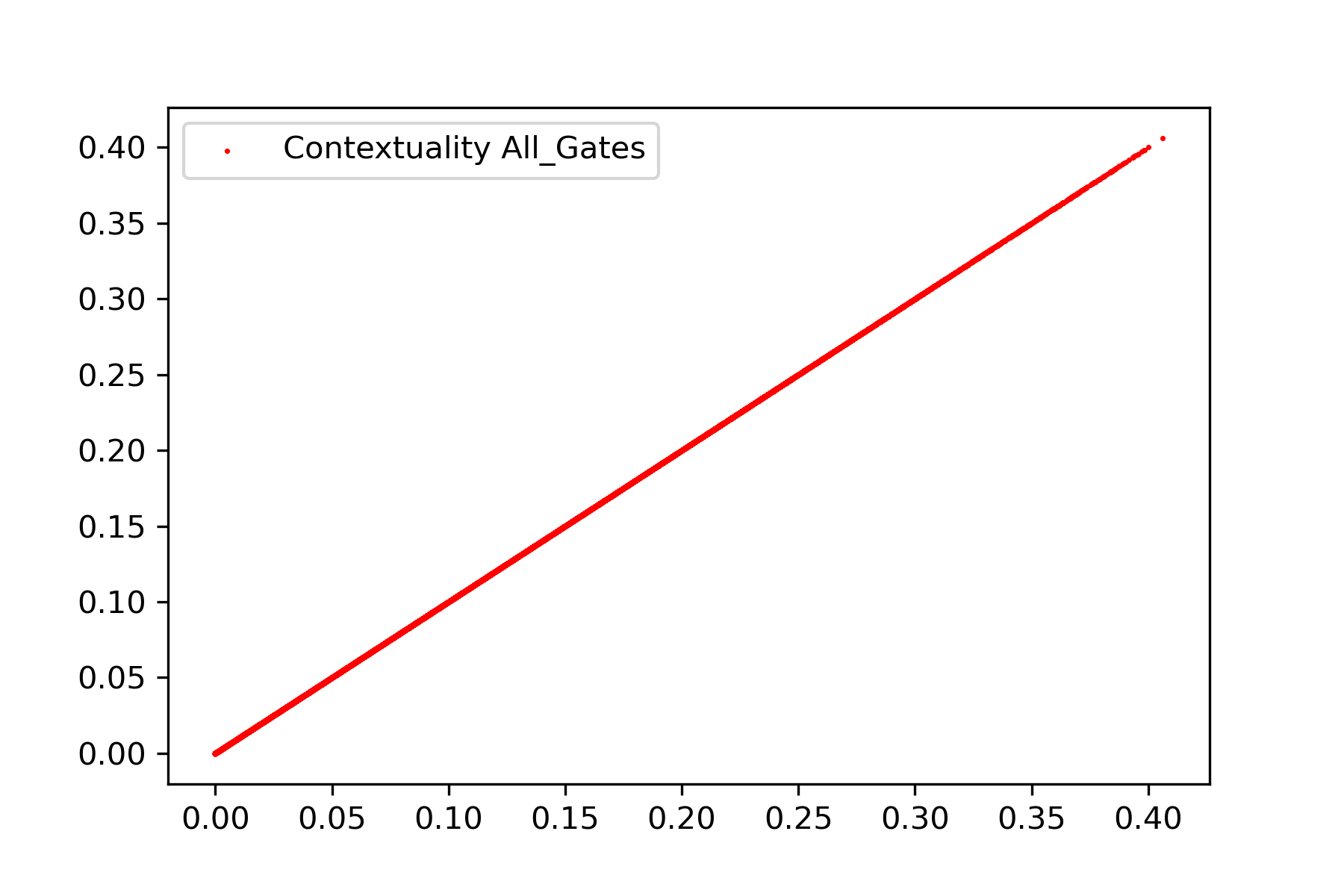}
         \caption{Contextuality values for all possible Qiskit gates.}
     \end{subfigure}
     \hfill
     \centering
     \begin{subfigure}{0.45\textwidth}
         \centering
         \includegraphics[width=\textwidth]{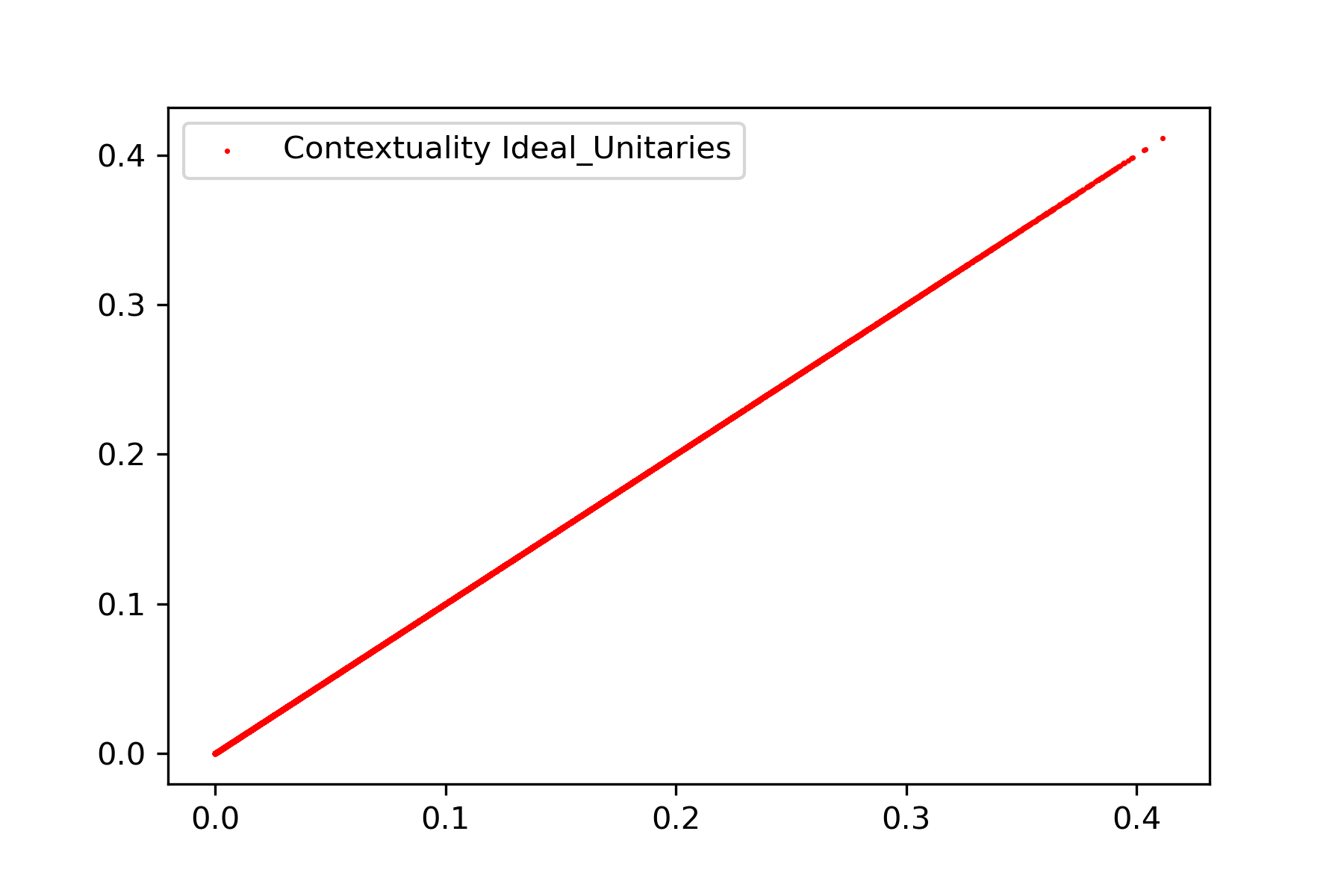}
         \caption{Contextuality values for 100000 unitary gates generated with the Haar measure.}
     \end{subfigure}
     \hfill
              \caption{For each set of gates we plot the contextuality values obtained vs themselves. This illustrates that the set of contextuality values of the non-universal set contains "holes", i.e., certain values are very difficult to approximate.}
       \label{f:Holes}
\end{figure}

At this point, it is interesting to ask: how much contextuality can be produced by a given set of quantum gates? In order to answer this question, we use a modified version of the function \verb"qiskit.circuit.random.random_circuit" of the Qiskit SDK \cite{Qiskit} to generate quantum random circuits with different sets of gates. We compute the contextuality of the states associated to $100000$ randomly generated quantum circuits of depth $200$ for two qubits. The first set is formed by the Clifford gates only. The second, contains also the $T$ gate. Finally, we use all possible gates available in the above mentioned Qiskit function. For completness, we also computed the contextuality of a set of states generated using radom unitaries (using the Python library \verb"SciPy" \cite{SciPy-NMeth}). The vast majority of the states thus generated shows no contextuality. We build histograms for the values obtained for those states with non-null contextuality. The results are depicted in Figure \ref{fig:Comparison_Qiskit}. We show the probability distributions associated to those histograms in order to see how states with non-null contextuality are distributed. As can be clearly seen, the obtained values are more distributed when the set of gates is universal (for example, for Clifford $+$ $T$). In order to quantify this, we compute the Shannon entropy associated to the distributions. We find that the entropy of Clifford $+$ $T$ is around fifteen times bigger than that of Clifford's alone. And that associated to all possible gates has the same order of magnitude than that of Clifford $+$ $T$. The message seems to be that, the richer the set of elementary gates employed, the distribution of the resource is more homogoenous among the quantum states generated. The distribution associated to the non-universal set is very ``picked", reflecting that the states produced do not cover the quantum state space in a reasonable way.

It is interesting to notice that, even if the Clifford set is non-universal, it can generate maximal contextuality (this corresponds to the pick observed on the right in Figure \ref{fig:Comparison_Qiskit} (a)). The reason for this should be intuitively clear: the Clifford set can generate states that maximally violate the CHSH inequalities. And these inequalities constitute, in turn, a contextuality scenario. The Gottesman-Knill theorem affirms that circuits generated by the Clifford set alone can be classically simulated \cite{Gottesman}. But some of the states thus generated are superposed and posses entanglement. For that reason, some authors argue that entanglement and superposition alone cannot be the source of the quantum speed-up (though see \cite{Cuffaro_Gottesman_Knill}). Following a similar reasomning line, our results suggest that contextuality alone --at least when quantified with the measure studied here-- cannot be the reason for the quantum speed-up either. But it seems that there is a clear difference between the distributions associated to universal vs non-universal sets of gates: the former are distributed in a more homogeneous way than the latter. Therefore, these results suggest that the quantum speed up might be related to how rich is the distribution of a resource for the states generated (see also the conceptual discussion presented in \cite{Holik_QP_and_QT_2022}). One might say that the contextuality values generated by the non-universal set display ``holes" in the quantum state space, meaning that some values cannot be reached (see Figure \ref{f:Holes}). Of course, these assertions are not conclusive, and will be studied with more detail in future works. But they show how useful the measure introduced in the previous section can be for a better understanding of quantum information problems.





\section{Conclusions}\label{s:Conclusions}

In this work we have elaborated on previous approaches and provided a very general definition of negative probability. Alike the Wigner function, our proposal does not relies on any Hilbert space structure. Thus, it provides a solid foundation for the description of quantum states based on measure theory, generalizing Kolmogorov's theory in a very natural way. Differently from previous approaches, it allows to describe infinite dimensional models.

For the discrete case, it is possible to use our definition to define a measure of quantum contextuality that can be easily computed numerically. We have used it to compute the contextuality associated to different scenarios. In particular, we computed the contextuality associated to Bell, PR boxes, and Mermin's box scenarios. The example of the Cat-like states illustrates that the contextuality values depend explicitly on the considered scenario. As an example, for Bell-type scenarios, these depend on the orientations of the angles of the spin observables. For that reason, for a given quantum state, we choose the angles that correspond to a maximal violation of the CHSH inequality.

The $L_{1}$-norm based contextuality measure turns out to be particularly useful to study how contextuality, understood as a resource, is distributed among the states generated by different sets of quantum gates. We find that the Clifford set can generate maximal contextuality. Using the Gottesman-Knill theorem, one could say that contextuality---quantified by the measure studied in this work--- on its own cannot be the reason for the quantum speed-up. Quite on the contrary, our results suggest that the main difference between universal vs non-universal sets of elementary gates is that the contextuality values of the former are distributed in a more homogeneous way than the later. This findings open the door for further inquiry, that we will address in future works.

\end{document}